\documentclass[runningheads]{llncs}
\usepackage[T1]{fontenc}
 \usepackage{amsmath}
\usepackage{amssymb}
\usepackage{algorithm}
\usepackage{wrapfig}
\usepackage[noend]{algpseudocode}
\usepackage[pdftex,dvipsnames]{xcolor}
\usepackage{orcidlink}
\usepackage{lineno}
\usepackage{tikzsymbols}
\usepackage{tikz}
\usepackage{mathtools}
\usepackage{thmtools}
\usepackage{thm-restate}
\usepackage{hyperref}
\usepackage[capitalise]{cleveref}
\usepackage{xspace}
\usepackage[new]{old-arrows}
\usepackage{float}
\usepackage{wrapfig}
\usepackage{caption}
\usepackage{subcaption}
\usepackage{graphicx}
\usepackage{mathabx}
\usepackage{stmaryrd}
\usepackage{wasysym}
\usepackage{pgfplots}
\usepackage{graphicx}
\usepackage{tikz}
\usetikzlibrary{automata,positioning,arrows.meta}
\usepackage{array}
\usepackage{longtable}
\usepackage{multirow}
\usepackage{thm-restate}

\renewcommand{\epsilon}{\varepsilon}
\renewcommand{\phi}{\varphi}
\usepackage[most]{tcolorbox}
\newtcolorbox{softbox}[1][Definition]{ 
  colframe=blue!20!white,      
  colback=blue!5!white,        
  colbacktitle=blue!15!white,  
  
  coltitle=black,              
  coltext=black,               
  
  fonttitle=\bfseries,         
  arc=5mm,                     
  boxrule=0.5pt,               
  toptitle=1mm, bottomtitle=1mm, 
  
  title={#1}
}
\newtcolorbox{simplebox}{
  colback=red!5!white,
  colframe=red!20!white,
  coltext=black,
  arc=3mm,
  boxrule=0.5pt,
  left=1.5mm, right=1.5mm,
  top=1mm, bottom=1mm
}
\usepackage[framemethod=TikZ]{mdframed}

\newcommand{\Artoo}{\textsf{R2D2}}
\newcommand{\Ceethree}{\textsf{C3PO}}

\renewcommand{\Circle}{\mathbin{\raisebox{-0.34ex}{\scalebox{1.5}{$\circ$}}}}

\newcounter{quest}[section]

\renewcommand{\Game}{\mathcal{G}}
\newcommand{\Players}{\Pi}
\newcommand{\Vertices}{S}

\newcommand{\Actions}{A}
\newcommand{\Av}{\mathsf{Av}}

\newcommand{\Terminals}{T}
\newcommand{\Dist}{\mathsf{Dist}}

\newcommand{\player}{p}

\newcommand{\action}{a}

\newcommand{\history}{h}

\newcommand{\strat}{\sigma}
\newcommand{\threshold}{t}
\newcommand{\initial}{s_0}

\newcommand{\Atoms}{Q}

\newcommand{\PlayerSet}{\varSigma}

\newcommand{\until}{\mathsf{U}}

\newcommand{\bstrat}{{\bar{\strat}}}
\newcommand{\baction}{{\bar{\action}}}
\newcommand{\bxi}{\bar{\xi}}

\newcommand{\sure}{\mathsf{sure}}
\newcommand{\almost}{\mathsf{almost}}
\newcommand{\limit}{\mathsf{limit}}
\newcommand{\ind}{\mathsf{ind}}
\newcommand{\Pre}{\mathsf{Pre}}
\newcommand{\sh}{\mathsf{sh}}

\newcommand{\team}[2]{\ensuremath{\langle \! \langle  {#1}  \rangle \! \rangle_{#2}}}
\newcommand{\teami}[2]{\ensuremath{\langle \! \langle {#1}  \rangle \! \rangle_{#2}^{\ind}}}
\newcommand{\teamr}[3]{\ensuremath{\langle \! \langle {#1}  \rangle \! \rangle_{#2}^{#3}}}

\newcommand{\opp}{\mathsf{O}}
\newcommand{\prob}{\mathbb{P}}
\newcommand{\etr}{\exists\Rb}

\newcommand{\Nb}{\mathbb{N}}
\newcommand{\Rb}{\mathbb{R}}

\newcommand{\Qb}{\mathbb{Q}}

\newcommand{\NP}{\mathsf{NP}}
\newcommand{\EXPSPACE}{\mathsf{EXPSPACE}}
\newcommand{\SAT}{\mathsf{SAT}}

\newcommand{\SQRTSUM}{\mathsf{SQRTSUM}}
\renewcommand{\P}{\mathsf{P}}

\newcommand{\coNP}{\mathsf{coNP}}
\newcommand{\PSpace}{\mathsf{PSPACE}}

\newcommand{\Ocomplexity}{\mathcal{O}}

\usepackage{xargs} 

\begin{document}
\title{Randomise Alone, Reach as a Team}
%
%
\author{
Léonard Brice\inst{1}\orcidlink{0000-0001-7748-7716} 
\and
Thomas A. Henzinger\inst{1}\orcidlink{0000-0002-2985-7724} 
\and
Alipasha Montaseri\inst{1}\orcidlink{0000-0003-3262-9332}
\and
Ali Shafiee\inst{1}\orcidlink{0009-0008-5947-1946}
\and
K. S. Thejaswini\inst{2}\orcidlink{0000-0001-6077-7514}}
 \authorrunning{L. Brice, T. A. Henzinger, A. Montaseri, A. Shafiee, and K. S. Thejaswini}
 \institute{Institute of Science and Technology Austria, Austria 
 \email{\{leonard.brice,tah, alipasha.montaseri,ali.shafiee\}@ista.ac.at}\\
  \and
 Université libre de Bruxelles, Belgium\\
 \email{thejaswini.raghavan@ulb.be}
 }

\maketitle            
\begin{abstract}
We study concurrent graph games where $n$ players cooperate against an opponent to reach a set of target states. Unlike traditional settings, we study distributed randomisation: team players do not share a source of randomness, and their private random sources are hidden from the opponent and from each other. 

We show that memoryless strategies are sufficient for the threshold problem (deciding whether there is a strategy for the team that ensures winning with probability that exceeds a threshold), a result that not only places the problem in the Existential Theory of the Reals ($\exists\mathbb{R}$) but also enables the construction of value iteration algorithms. We additionally show that the threshold problem is $\NP$-hard. For the almost-sure reachability problem, we prove $\NP$-completeness.

We introduce Individually Randomised Alternating-time Temporal Logic (IRATL). This logic extends the standard ATL framework to reason about probability thresholds, with semantics explicitly designed for coalitions that lack a shared source of randomness. 
On the practical side, we implement and evaluate a solver for 
the threshold and almost-sure problem 
based on the algorithms that we develop.
\end{abstract}


\section{Introduction}
Multi-agent systems provide a robust mathematical framework for modelling complex interactions across a diverse array of domains. From the formal verification of reactive systems~\cite{FKNP11,MultAgentVerification96} and cyber-physical architectures~\cite{SEC16,LFB20} to the study of epidemic processes~\cite{Lef81} and distributed probabilistic programs~\cite{dAHJ01}, the interplay between autonomous entities is central to ensuring system correctness.
In these models, such entities are typically represented as \emph{players}, repeatedly and concurrently choosing among some set of \emph{actions}. While classical verification often considers a single controller against an adversarial environment, there are plenty of applications that demand multi-agent models with access to varying degrees of cooperation and information.

A traditional tool to capture what coalitions of players can achieve in such settings is Alternating Temporal Logic (ATL)~\cite{AHK02}.
ATL captures specifications of the following form: ``Can a coalition achieve a given temporal objective against all strategies of the other players?''
In such a setting, one need only consider deterministic strategies, since randomising among several actions does not bring any additional power.
Important extensions, however, are Randomised ATL (RATL)~\cite{AHK07} and Probabilistic ATL (PATL)~\cite{CL07}, both of which capture notions such as almost-sure and limit-sure winning---and the latter considers threshold winning, that is, guaranteeing that the probability of winning exceeds a given threshold $\threshold$. 
For those types of specifications, randomised strategies are strictly more powerful than deterministic ones.
Nevertheless, the semantics of RATL and PATL allow coalitions to randomise their actions jointly, implicitly assuming that they can resort to shared sources of randomness, or equivalently, that they can communicate about their random choices privately.
Such a hypothesis simplifies significantly the setting, since the coalition can then be seen as a single player; but it is not realistic in many real-life distributed and multi-agent systems, where challenges often lie, precisely, in the inability of the agents to communicate.
This paper aims at bridging this gap.

\paragraph{An example.}
Picture an object that needs to be moved to the other side of a sliding door. 
Two robots, named $\Artoo$ and $\Ceethree$, can
choose at each step to move the object either to the left or to the right,
while an adversarial environment moves the sliding door, such that one among the left and the right side is open at each time step. 
If both $\Artoo$ and $\Ceethree$ choose the same direction as the environment, then the object moves to the other side and they win. If  $\Artoo$ and $\Ceethree$ choose the same side, but the environment opens the other side
then the object remains stationary, and they have to try again. If they choose different sides,
the object breaks due to the opposing forces.
This situation is illustrated in \Cref{fig:firstexample}.
With what probability can $\Artoo$ and $\Ceethree$ ensure reaching $s_\mathsf{goal}$?

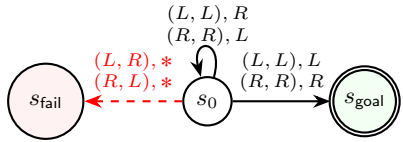
\begin{figure}
    \centering
\begin{tikzpicture}[
    >=Stealth, node distance=1.3cm, auto,
    state/.style={circle, draw, minimum size=0.5cm, thick, align=center},
    decision/.style={circle, draw, fill=blue!5, minimum size=0.7cm, thick},
    goal/.style={circle, draw, double, fill=green!5, minimum size=0.7cm, thick},
    fail/.style={circle, draw, fill=red!5, minimum size=1cm, thick}
]

    \node[state] (s0) {$s_0$};
    \node[goal, right =of s0] (sgoal) {$s_\mathsf{goal}$};
    \node[fail,  left=of s0] (sfail) {$s_{\mathsf{fail}}$};
    
    \draw[->, thick] (s0) -- node[sloped, above, font=\scriptsize, text width=4cm, align=center] {
     $(L,L),L$ \\$(R,R),R$
    }  (sgoal);

       \draw[->, thick, red, dashed] (s0) -- node[sloped, above, font=\scriptsize,text width=4cm, align=center] {$(L,R),*$ \\$(R,L),*$} (sfail);
       \draw[->, thick] (s0) to [out=75,in=105,looseness=8] (s0);
       \path[loop above, thick] (s0) to node [above, yshift=0.2cm, font=\scriptsize] {$(L,L),R$ } node[above,yshift=-0.1cm, font=\scriptsize] {$(R,R),L$} (s0);
\end{tikzpicture}
\caption{The two agents choose among actions $\{L,R\}$ and the environment chooses among actions $\{L,R\}$.}
    \label{fig:firstexample}
\end{figure}

Consider first the setting where $\Artoo$ and $\Ceethree$ have access to a shared source of randomness, which can't be observed by the adversary. Such games can be viewed as standard two-player concurrent games, where the team acts as a single meta-player.
By playing the joint actions $(L, L)$ or $(R, R)$ with probability $\frac{1}{2}$ each,
$\Artoo$ and $\Ceethree$ ensure that for any environmental choice, the probability of advancing
is $\frac{1}{2}$, and the probability of reaching the failure state is $0$. 
By repeating this, 
it is almost sure that the target will eventually be reached.

However, when $\Artoo$ and $\Ceethree$ do not share a source of randomness nor a private communication channel,
they have to randomise their actions independently.
One can show that against any such strategy profile, the environment can respond in such a way that the probability of eventually reaching the target does not exceed $\frac{1}{3}$. 
Observe here that the severity of this result relies on the order of the quantifiers: we assume that the team fixes a collective strategy first, and that the environment responds to it, hence the value $\frac{1}{3}$ is what is called in the literature the \emph{max-min value}.
If, instead, the environment must commit to a strategy first, then $\Artoo$ and $\Ceethree$ can respond deterministically and reach the target state almost surely (with probability $1$), which is called the \emph{min-max value}. It is because the team can always coordinate on a single deterministic action (e.g., both constantly playing $L$ or $R$) that the environment plays with positive probability. By this, they can avoid the failure state and ensure they eventually progress.
We assume the burden of pessimism and focus on the max-min value.


\paragraph{Our Results.}
We first start with the following problem for team games.
\begin{simplebox}
    \textbf{Threshold problem:} 
    Given a game, 
    given a rational threshold $t\in[0,1]$, does there exist a \emph{collective} strategy for the team 
    that ensures the target is reached with probability strictly greater than $t$?
\end{simplebox}
Our first main result establishes that when there exists such a collective strategy, there exists a memoryless one (\Cref{thm:memmless-optimal}).
This implies that these games can be fully characterised by local, state-dependent strategies. 
Therefore, given an instance of the threshold problem, we can construct a formula in the Existential Theory of the Reals (ETR) that is satisfiable if and only if we have a positive instance.
This proves that the threshold problem 
lies in the complexity class $\exists\Rb$ (itself contained in $\PSpace$~\cite{Can88,SCM24}). 
We also prove that the problem is $\NP$-hard using a reduction from $k$-clique problem. This is unlike the situation for two-player concurrent games, where no such hardness results are known. In fact, for two-player concurrent games (or for games where the team can share randomness, which reduces to the latter) the best known lower bound for the threshold problem is $\SQRTSUM$-hardness~\cite{EY15} and $\P$-hardness derived from turn-based reachability games~\cite{Imm81} and the best known upper bound is the complexity class $\exists\Rb$~\cite{CDH12}.
While our encoding into  ETR provides an algorithm, such large logical sentences are notoriously difficult for existing solvers to solve. 
To provide a more computationally accessible approach, we also develop a Value Iteration (VI) algorithm that converges to the optimal max-min value---even though it might never reach it.

Secondly, we prove $\NP$-completeness of almost-sure reachability.
\begin{simplebox}
    \textbf{Almost-sure problem:} 
    Given a game, does there exist a collective strategy for the team that ensures that the target is reached with probability~$1$?
\end{simplebox}

While this problem lies in $\P$ in a two-player setting,
we show that it is $\NP$-hard even with three players.
Membership in $\NP$ is obtained by proving that memoryless strategies also suffice in this setting---but the proof requires techniques different from that of the threshold problem.

Thirdly, we implement solvers for the threshold and almost-sure problems, and evaluate them on existing benchmarks that we modify to fit our setting. 
For the threshold problem, we first encode the entire game as a single formula in $\exists\Rb$ and use SMT solvers; while this offers strong theoretical guarantees, the formula's complexity often leads to timeouts. In contrast, we employ VI algorithms using various termination criteria and heuristics. These heuristics provide \emph{sound under-approximations}, and terminate with results close to the actual max-min values in our benchmarks. Because VI methods query the $\exists\Rb$ solver using smaller formulas, they terminate significantly faster than solving the whole encoding. For the almost-sure case, given its $\NP$-completeness, we utilise $\SAT$ solvers on the succinct $\SAT$ encodings that we construct.
We use state-of-the-art PRISM-games to solve the threshold and the almost-sure version of the strict subcase of the games with shared randomness as a baseline. 
Despite solving a computationally harder game, our solvers achieve runtimes comparable to PRISM-games~\cite{KN0S20}.

Finally, we propose a new logic called \emph{Individually Randomised ATL (IRATL)} to capture the inability of a coalition to randomise collectively. The syntax of IRATL extends ATL 
with operators capturing the fact that a coalition can enforce an objective without using any shared randomness.
For example, considering our robot team, we would write the formula $\teamr{\Artoo, \Ceethree}{\almost}{\mathsf{sh}} (\Diamond\{s_{\mathsf{goal}}\})$ to check if the team can reach $s_{\mathsf{goal}}$ almost-surely using \emph{shared} ($\mathsf{sh}$) randomness. In contrast, to query if the team can reach the target with a probability strictly greater than $0.3$ using \emph{independent} ($\mathsf{ind}$) sources of randomness, we write $\teami{\Artoo, \Ceethree}{> 0.3} (\Diamond \{s_{\mathsf{goal}}\})$. As established previously, both sentences hold true in their respective settings. We effectively solve the model-checking problem for a key fragment of this logic.

\subsubsection{Related work}
\paragraph{Logics for stochastic team games.} 
Alternating-time Temporal Logic (ATL) was first presented in the work of Alur et al.~\cite{AHK02}, and has inspired a rich ecosystem of extensions.
While the ATL semantics considers only deterministic moves, several of those extensions introduce randomness into the picture, such as Randomised ATL (RATL)~\cite{AHK07}, PATL~\cite{CL07}, PAMC~\cite{pamc}, rPATL~\cite{Chen2012}, and SGL~\cite{10.1007/s00236-012-0156-0}.
SGL and rPATL are defined only on turn-based games, where the distinction between shared and individual randomisation is not relevant.
RATL, PATL, and PAMC assume that players in a coalition can resort to shared randomisation, making it possible to reduce the setting to a standard two-player zero-sum game.\footnote{PATL and PAMC were actually defined without this hypothesis, but a careful reading of the proofs in the cited papers shows that it is then implicitly introduced, perhaps inadvertently. This ambiguity has been propagated across the subsequent literature: all results seem to assume shared randomisation, but the definitions are sometimes consistent with that hypothesis~\cite{10.1007/978-3-642-28756-5_22} and sometimes not~\cite{zhangpang2010}.
Let us also note that~\cite{CL07} claims that model-checking for PATL is in $\NP \cap \coNP$ by invoking a reduction from LTL objectives to parity ones; however, such a reduction is not polynomial-time.}

Strategy logic~\cite{CHP10} is another well-known extension of ATL, that considers only deterministic strategies.
Its extension Probabilistic Strategy Logic (PSL)~\cite{probstratlog} is very expressive, and since strategies are represented as variables, individual randomisation is implicitly implied.
But this expressivity comes at the cost of undecidability in the general case, and very high complexities even when the players are restricted to memoryless strategies: the easiest known fragment that subsumes our setting is only known to be in $\EXPSPACE$.

Finally, probabilistic Resource-Bounded ATL (pRB-ATL, \cite{pRBATL}) is another logic that subsumes our setting, with a broader semantics including limited resources for the agents.
No complexity upper bound is known for pRB-ATL model-checking.

\paragraph{Other multiplayer game settings.}
Several works including those on population games~\cite{BBLM25}, concurrent parameterised games~\cite{BBM19}, and synthesis of coalition strategies~\cite{BBM20}, work on settings with an unknown number of players as opposed to the fixed number of players in our case. Fixing the number of players enables us to obtain value-iteration-like algorithms, more amenable to implementation. A recent work on one-shot games with partially shared randomness (various sources of randomness, each of them being accessible to a given subset of players) shows that the threshold problem with a non-strict inequality (instead of strict, as in our setting) is $\exists\Rb$-complete~\cite{BHT26} for one-shot games. However, neither the hardness nor the easiness results extend to our setting.
\paragraph{Fixed environments and multi-agent learning.} 
As a side observation, we consider settings where the environment is memoryless (modelling fixed but unknown stochastic process). 
This  setting reduces to multi-agent reinforcement learning, where agents must learn to coordinate in a stochastic environment~\cite{Gue02,ZYB20,KV06}. 
\paragraph{Tools.} The landscape of automated verification for multi-agent systems is rich with tools, yet, to the best of our knowledge, none natively support the independent randomisation constraint central to our work. The most prominent tool, PRISM-games~\cite{CFKPS13,KN0S20}, extends probabilistic model checking to the two-player version. While PRISM-games can verify properties for a coalition of players (e.g., using the logic rPATL), it treats it as a single ``meta-player'' capable of correlated strategies. Tools such as MCMAS~\cite{LQR17} and EVE~\cite{EveTool} enable strategic reasoning via ATL and Nash equilibrium synthesis. However, MCMAS focuses primarily on qualitative (Boolean) verification and epistemic properties, while EVE and related Nash solvers (e.g., PRISM-Nash) target non-zero-sum equilibria.

\section{Preliminaries}

\subsection{Probabilities}
Given a finite set $S$, a \emph{probability distribution} over $S$ is a mapping $d: S \to [0,1]$ that satisfies the equality $\sum_{x \in S} d(x) = 1$.
The set of probability distributions over $S$ is denoted by $\Dist S$.

\subsection{Games}

We use the word \emph{game} for team concurrent stochastic reachability games.

\begin{definition}[Game structure, game]
    A \emph{game structure} is a tuple $G = (\PlayerSet, \Vertices, (\Actions_p)_{p \in \PlayerSet}, (\Av_p)_{p \in \PlayerSet}, \delta)$, 
    consisting of:
    \begin{itemize}
        \item a set $\PlayerSet$ of \emph{players};

        \item a set $\Vertices$ of \emph{states};

        \item for each player $p$, a set $\Actions_p$ of \emph{actions};

        \item for each player $p$ and state $s$, a set $\Av_p(s) \subseteq \Actions_p$ of actions available at state $s$ (always assumed to be nonempty);

        \item a \emph{transition function} $\delta: \Vertices \times \prod_{p \in \PlayerSet} \Actions_p \to \Dist(\Vertices)$ that maps each tuple $(s, (a_p)_p)$ with $a_p \in \Av_p(s)$ for each $p$ to a probability distribution over $\Vertices$.
    \end{itemize}
    
    A \emph{game} is a tuple $\Game = (G, \Players, T, \initial)$, that consists of a game structure equipped with a \emph{team} $\Players \subseteq \PlayerSet$, a set $T \subseteq \Vertices$ of \emph{target states}, and an \emph{initial state} $\initial \in \Vertices$.
\end{definition}







When referring to a game $\Game$ or a game structure $G$, we often use the notations $\Vertices$, $\delta$, etc., without recalling them.
Given a set $P \subseteq \PlayerSet$ and a vertex $s$, we often write $\Av_P(s)$ (or $\Av(s)$ if $P = \PlayerSet$) to denote the product $\Av_P(s) = \prod_{p \in P} \Av_p(s)$.
For technical reasons, we assume that the target set $T$ is \emph{absorbing}, i.e. that for every state $t \in T$ and action profile $\baction$ we have $\delta(t, \baction)(t) = 1$.

We call \emph{play} (resp. \emph{history}) in the game structure $G$ an infinite (resp. finite) word over the alphabet $S$.
We sometimes say that the team \emph{wins} the play $\pi$ if that play contains a target state $s \in T$, and \emph{loses} it otherwise.

\subsection{Strategies, and strategy profiles}

A \emph{strategy} for player $\player \in \PlayerSet$ is a mapping $\strat_\player$ that maps each history $\history s$ to a probability distribution $\strat_\player(\history s) \in \Dist\, \Av_\player(s)$.
A \emph{strategy profile} for the set $P \subseteq \PlayerSet$ is a tuple $\bstrat_P = (\strat_\player)_{\player \in P}$, where each $\strat_\player$ is a strategy for player $\player$.
A \emph{collective strategy} is a strategy profile $\bstrat_\Players$ for the team.
A \emph{complete} strategy profile is a strategy profile $\bstrat_{\PlayerSet}$ for $\PlayerSet$, often written $\bstrat$, without subscript.
Given two strategy profiles $\bstrat_P$ and $\bar{\tau}_Q$ with $P \cap Q = \emptyset$, we write $(\bstrat_P, \bar{\tau}_Q)$ for the strategy profile $\bar{\eta}_{P \cup Q}$, where $\eta_\player = \strat_\player$ if $\player \in P$, and $\tau_\player$ if $\player \in Q$.
Given two strategy profiles denoted by $\bstrat_P$ and $\bstrat_{\PlayerSet \setminus P}$, we simply write $\bstrat$ for the complete strategy profile $(\bstrat_P, \bstrat_{\PlayerSet \setminus P})$.
Given a strategy profile $\bstrat_P$ and a history $h$, we write $\bstrat(h)$ for the induced distribution over joint actions, i.e., for the distribution $\bstrat(h): \baction_P \mapsto \prod_{p \in P}  \strat_p(h)(\action_p)$.

A complete strategy profile $\bstrat$ defines a probability distribution $\prob_\bstrat$ over  plays, where for each history $hs$ and state $t$, we have:
$$\prob_\bstrat(hst S^\omega) = \sum_{\baction \in \Av(s)} \bstrat(hs)(\baction) \cdot \delta(s, \baction)(t).$$

A strategy $\strat_p$ is \emph{deterministic} if for every history $h$, there exists an action $a$ such that $\strat_p(h)(a) = 1$.
It is \emph{memoryless} if for every state $s$ and every history $h$, we have $\strat_p(hs) = \strat_p(s)$.

\subsection{Problems}

Throughout this paper, our goal is to design collective strategies for teams that guarantee a certain property against every strategy of the other players.
We always assume that the team players fix their strategies first, and the other players (the \emph{opponents}) respond to it.
From the opponents' perspective, responding to a collective strategy amounts to finding optimal strategies in a Markov decision process.
It is a well-known result that deterministic strategies are sufficient in such a setting.
Consequently, contrary to the team players, the fact that the opponents do not share a common source of randomness is not a restriction, and the opponents can be considered as a single player.
For convenience, we will therefore always assume that the set $\PlayerSet \setminus \Players$ is a singleton $\{\opp\}$, where player $\opp$ is called \emph{opponent}.
In such a setting, a natural question is to ask whether the team players can win the game with a probability greater than a given threshold.

\begin{problem}[Threshold problem]\label{problem:threshold}
    Given a game $\Game$ and a threshold $\threshold \in \Qb \cap [0, 1]$, does there exist a collective strategy $\bstrat_\Players$ such that for every strategy $\strat_\opp$, the set $\Terminals$ is reached with probability strictly greater than $t$?
\end{problem}

A less ambitious question is to ask whether they can win with probability~$1$.

\begin{problem}[Almost-sure problem]\label{problem:almostSure}
    Given a game $\Game$, does there exist a collective strategy $\bstrat_\Players$ such that for every strategy profile $\bstrat_\opp$, the set $\Terminals$ is reached with probability 1?
\end{problem}

In the two next sections, we design algorithms for those two problems.

\section{The threshold problem}\label{sec:Threshold}
 In this section, we analyse the structure of the team's strategies and establish the theoretical foundations for solving Problem~\ref{problem:threshold}. 
  The main result of this section is \cref{thm:memmless-optimal}, which establishes that memoryless strategies are sufficient. 

 Building on this result, we provide a decision procedure by encoding the threshold problem as an ETR formula.
 Finally, we end the section with a result on $\NP$-hardness for the threshold problem, which adds to the $\SQRTSUM$-hardness result that is already known from the two-player case~\cite{EY15}.

\subsection{Optimality of memoryless strategies} 

\label{sec:OptimalityofMemorylessStrategies}

In this subsection, we prove that memoryless strategies suffice to exceed a given threshold.
The structure of the proof is similar to that of the same result in standard two-player zero-sum concurrent games~\cite{CDH12}.
However, the latter heavily relies on the equality of the max-min and min-max values, which no longer holds in our setting.
Our proof requires therefore more technical effort.

\subsubsection{Value iteration.}
As a step toward our result, and as an algorithmic tool, we define a value iteration sequence.
A \emph{valuation} is a mapping \( v : \Vertices \to [0, 1] \).
We define the \emph{predecessor operator} based on this intuition: for a fixed valuation $v$ and a state $s$, we consider a \emph{local one-shot game} played in $s$. In this local game, the team players choose a joint distribution over their available actions, and the opponent responds with an action as well.
If $t$ is the subsequent state, then the team gets \emph{payoff} $v(t)$.
The team players want to maximise the expected payoff, while the opponent wants to minimise it. 
To avoid confusion, strategies in this one-shot game are \emph{selectors}.
The predecessor operator is defined in three steps: given fixed strategies for both the team and the opponent, given a fixed team strategy against an optimal opponent, and finally, the optimal one-step value.

First, given a valuation $v$, a collective selector $\bxi_{\Players}$, and an opponent selector $\xi_{\opp}$, we define the expected payoff at state $s$ as:
$$
    \mathsf{Pre}_{\bxi}(v)(s) 
    \;=\; 
    \sum_{\baction \in \Av_{\Players}(s)} \sum_{b \in \Av_{\opp}(s)} \sum_{t \in \Vertices} 
    \bxi_{\Players}(s)(\baction) \cdot \xi_{\opp}(s)(b) \cdot \delta(s, \baction, b)(t) \cdot v(t).
$$
We now define the value guaranteed by a fixed collective selector $\bxi_{\Players}$:
\[
    \mathsf{Pre}_{\bxi_{\Players}}(v)(s)
    \;=\;
    \inf_{\xi_{\opp} \in \Lambda_{\opp}(s)}
    \mathsf{Pre}_{\bxi}(v)(s)
\]
where $\Lambda_{\opp}(s) = \Dist(\Av_\opp)$ is the set of selectors for the opponent in state $s$.

Finally, we define the \emph{predecessor operator} $\mathsf{Pre}$, which represents the maximum value the team can guarantee in one step against the opponent:
\[
    \mathsf{Pre}(v)(s)
    \;=\;
    \sup_{\bxi_{\Players} \in \Lambda_{\Players}(s)}
    \inf_{\xi_{\opp} \in \Lambda_{\opp}(s)}
    \mathsf{Pre}_{\bxi}(v)(s)
\]
where $\Lambda_{\Players}(s)$ is the set of collective selectors for the team in state $s$.

We can then define a sequence $(v_n)_n$ of valuations as follows: the valuation $v_0$ maps each target state to the value $1$ and each other state to the value $0$, and for each $n$, we have $v_{n+1} = \Pre(v_n)$.
We then have the following result.

\begin{restatable}[App.~\ref{app:VIconvergencetovalue}]{lemma}{lmVIconvergencetovalue}\label{lm:VIconvergencetovalue}
    For every vertex $s$, the sequence $(v_n(s))_n$ converges to the max-min value of the game, when initialised in $s$.
\end{restatable}

As we will show in \Cref{sec:experiments}, this value iteration sequence also entails a sound under-approximation algorithm. For now, let us use this tool to prove our result.

\subsubsection{Memoryless strategies are sufficient.}

\begin{restatable} [App.~\ref{app:memmless-optimal}]{theorem}{thmmemlessoptimal}\label{thm:memmless-optimal}
    Let $\Game$ be a game and let $\threshold$ be a threshold.
    If there is a collective strategy $\bstrat_\Players$ that guarantees victory with probability strictly greater than $\threshold$ against every opponent's strategy, then there is a memoryless one.
\end{restatable}

\begin{proof}[sketch]
Using \Cref{lm:VIconvergencetovalue}, there exists an index $n$ such that the $n$th valuation, on $\initial$, takes a value that is greater than $t$.
We assign a ``rank'' to each state corresponding to the iteration index where its value was determined. 
The memoryless strategy is then constructed to ensure that at each step, the next state has either a higher valuation or a lower rank, which guarantees convergence to the states with valuation $1$ and with rank $0$: the target states.
\qed\end{proof}


\subsubsection{A variant: playing against a memoryless opponent.}\label{sec:memlessOpp}

Before proceeding to the study of the threshold problem, we make a remark about a variant of our problem: the case where the opponent is restricted to memoryless strategies.
This modelling choice is well-motivated by the physical reality of autonomous systems, 
whose environments are often fixed (but unknown) stochastic processes rather than strategic adversaries. 
We show here that \Cref{thm:memmless-optimal} does not extend to this setting: surprisingly, restricting the opponent to memoryless strategies can force the team players to use memory to play optimally.
The following example illustrates that solving this variant would require significantly different techniques from the field of multi-agent learning. 


\begin{restatable}[App.~\ref{app:memlessOpp}]{theorem}{memlessOpp}\label{thm:memlessOpp}
    There exists a game $\Game$ 
    such that against an opponent restricted to memoryless strategies, the team can achieve a strictly higher probability of reaching the target using a strategy with memory than with any memoryless collective strategy.
\end{restatable}
\begin{proof}[sketch]
We give an example in the game in \cref{fig:memory}. 
Intuitively, the team players can agree to wait arbitrarily long to \emph{learn} the probability distribution during the waiting stage and then play the appropriate action to maximise the probability of reaching the target state.
\qed\end{proof}

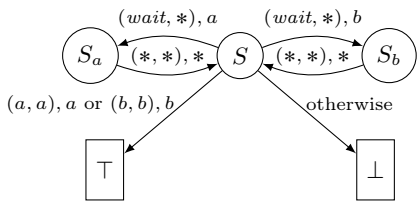
\begin{figure}
    \centering
    \begin{tikzpicture}[>=latex, node distance=1.3cm, auto,
      state/.style={draw, circle, minimum width=0.5cm, font=\small},
      terminal/.style={draw, rectangle, minimum width=0.5cm, minimum height=0.8cm, font=\small},
      lab/.style={font=\scriptsize}]
    
      \node[state] (s) {$S$};
      \node[state, left=of s] (sa) {$S_a$};
      \node[state, right=of s] (sb) {$S_b$};
    
      \node[terminal, below=of s, xshift=-1.8cm,yshift=0.5cm] (t) {$\top$};
      \node[terminal, below=of s, xshift=1.8cm,yshift=0.5cm] (trap) {$\perp$};
    
      \draw[->] (s) to[bend right=20] node[lab, swap] {$(\textit{wait},*), a$} (sa);
      \draw[->] (s) to[bend left=20] node[lab] {$(\textit{wait},*), b$} (sb);
      
      \draw[->] (sa) to[bend right=20] node[lab, above] {$(*, *), *$} (s);
      \draw[->] (sb) to[bend left=20] node[lab, above] {$(*, *), *$} (s);
    
      \draw[->] (s) -- node[lab, left, pos=0.4] {$(a,a), a \text{ or } (b,b), b$} (t);
    
      \draw[->] (s) -- node[lab, right, pos=0.4] {$\text{otherwise}$} (trap);
    
    \end{tikzpicture}
    \caption{An example game for opponent with no memory.}
    \label{fig:memory}
\end{figure}

\subsection{Complexity bounds} 
\label{sec:ETReasiness}

After a short detour, we return to the setting where the opponent has (potentially) unbounded memory.
We establish complexity  bounds.

\subsubsection{Membership in the class $\etr$.}
We first show that the problem is in $\exists\Rb$ by constructing an equivalent ETR formula. Our construction relies on \cref{thm:memmless-optimal}.
\begin{restatable}[App.~\ref{app:ETRComplete}]{theorem}{thmETRComplete}\label{thm:thmETRComplete}
    The threshold problem lies in the class $\exists\Rb$.
\end{restatable}
\begin{proof}[sketch]
To help our reduction, we introduce $\lambda$-discounted games. 
These games share the same structure as concurrent games but employ a discounted payoff function parameterised by a \emph{discount factor} $\lambda \in (0,1)$.
These games are \emph{quantitative}: if the team reaches the target set $\Terminals$ after $k$ steps, the team gets the payoff $\lambda^k$ (instead of $1$).
We construct a polynomial-size ETR formula $\Psi(\lambda)$ which is satisfiable if and only if there exists a collective strategy such that the $\lambda$-discounted value exceeds $\threshold$. Since this only captures the discounted version of the problem, we then show that it is equivalent to its qualitative counterpart, in the sense that there exists $\lambda$ such that the formula $\Psi(\lambda)$ is satisfiable if and only if we have a positive instance of the threshold problem.
 
This intermediary reduction to the $\lambda$-discounted setting is necessary. A direct encoding of the reachability problem faces a hurdle: the Bellman optimality equations for reachability generally admit multiple fixed points. Consider, for example, a simple loop $s\rightarrow s$ with no path to the target. The true reachability probability is $0$.
However, the assignment $v(s) = 1$ trivially satisfies the Bellman equations one would write, 
making it a valid but spurious fixed point. One standard solution is to impose the restriction that it is the least fixed point.
However, enforcing a ``least fixed point'' constraint would require the $\min$ operator, which is not permitted in the standard $\exists\Rb$ framework. 
Alternatively, identifying the non-zero support set requires guessing (if there exists a subset of vertices from which the target can be reached with nonzero probability), which puts the problem in a larger complexity class  $\NP^{\exists\Rb}$. This makes our intermediary detour via $\lambda$-discounted settings necessary.
\qed\end{proof}

\subsubsection{$\NP$-hardness.} 
We finally end with a proof of $\NP$-hardness for this problem. Note that it is also $\SQRTSUM$-hard, following from the two-player case~\cite{EY15}.
For the sake of completeness, we also note that the threshold problem with non-strict inequalities was recently shown to be complete for the class $\exists\Rb$~\cite[Lemma 5.7]{BHT26}.

\begin{restatable}[App.~\ref{app:NPhardnessThreshold}]{theorem}{thmNPhardnessThreshold}\label{thm:NPhardnessThreshold}
    The threshold problem is $\NP$-hard, even with three states and three players.
\end{restatable}
\begin{proof}[sketch]
    We establish $\NP$-hardness via a reduction from the \emph{clique problem}. Given a graph $G$ and an integer $k$, we construct a game with two team players and three states: an initial state, a target state, and a trap state.
The team players select actions consisting of a tuple $(u, i) \in V \times \{1, \dots, k\}$, claiming that they have a clique whose $i$-th vertex is $u$. Simultaneously, the opponent selects a pair of indices $(i, j)$ to verify. Transitions from $\initial$ are defined as follows:
the play reaches the target if the team plays $((u, i), (v, j))$ and the opponent selects $(i, j)$ such that the selection is consistent with a clique ($u=v$ if $i=j$, and $(u, v) \in E$ if $i \neq j$).
 The game transitions to the trap state if the team matches the opponent's indices $(i, j)$ but fails the clique consistency check ($u \neq v$ when $i=j$, or $(u, v) \notin E$ when $i \neq j$).
 If the team's chosen indices do not match the opponent's query, the game loops back to repeating the round.
We show that if $G$ contains a clique of size $k$, the team can reach the target state with probability $1$. Otherwise, the maximum probability of reaching the target is bounded by $\frac{3}{4}$. 
    The threshold problem is therefore $\NP$-hard, even with threshold $\frac{3}{4}$.
\qed\end{proof}

\section{The almost-sure problem}\label{sec:almostSure}
In this section, 
we focus on Problem~\ref{problem:almostSure}: deciding whether the team can win almost surely.  
We show that this problem is $\NP$-complete.
First, in \Cref{sec:almost-sure-memoryless}, we show that memoryless strategies are also optimal for this problem. 
 Second, in \Cref{sec:almost-sure-complexity}, we prove the upper and the lower bound.

\subsection{Optimality of memoryless strategies}
\label{sec:almost-sure-memoryless}

In this subsection, 
we show that if the team can guarantee reaching the target almost-surely, 
it can do so using a memoryless collective strategy. Formally, we show the correctness of the following theorem. Note that it is not a consequence of \cref{thm:memmless-optimal}, since the threshold problem is defined with strict inequalities.

\begin{restatable}[App.~\ref{app:almostsurememless}]{theorem}{almostsurememless}\label{thm:almostsurememless}\label{thm:almost-sure-memless}
    Let $\Game$ be a game.
    If there is a collective strategy $\bstrat_\Players$ that guarantees victory with probability~$1$ against every opponent's strategy, then there is a memoryless one.
\end{restatable}

\begin{proof}[sketch]
    We assign a rank (a value from $\Nb\cup \{\infty\}$) to each vertex, that captures ``distance'' from the target: the target set $T$ has rank $0$,
and the states that have rank $k+1$ are exactly those where the team can force a transition to states with rank $k$ with positive probability, without visiting a vertex of infinite rank. 
We construct a memoryless collective strategy as follows.
At every state, the team plays to guarantee a positive probability of moving to a lower rank, while staying within the winning set (the states with finite rank).
Since the probability of reducing the rank is always strictly positive, 
the game is forced to eventually progress and reach rank 0 (the target) with probability~1.
\end{proof}

\subsection{Complexity bounds}\label{sec:almost-sure-complexity}
\subsubsection{SAT Encoding for $\NP$ membership.}
\label{sec:almost-sure-sat}
We show the following theorem.

\begin{restatable}[App.~\ref{app:almostsuresat}]{theorem}{almostsureeasiness}\label{thm:almostsureeasiness}
    The almost-sure problem is in $\NP$.
\end{restatable}

\begin{proof}[sketch]
    Based on \cref{sec:almost-sure-memoryless}, we know that if an almost-sure winning collective strategy exists, there exists a memoryless one. We then give an explicit $\SAT$ formula $\Phi$ that is satisfiable if and only if there exists 
    a memoryless strategy that guarantee almost-sure reachability.
The construction of $\Phi$ 
relies on the fact that almost-sure reachability is a qualitative property: 
it depends only on the \emph{support} of the strategies, not on the precise probability values. 
The formula $\Phi$ encodes the search for a winning set $W$ and the set of support of strategy
for each player. 
The constraints enforce the existence of the ranking function as defined in the proof of \cref{app:almostsurememless}. We also apply some optimisations to reduce the number of variables, like using a binary encoding of the rank. 
\qed\end{proof}

We end this section with a hardness result.
The following theorem is an immediate consequence of the construction presented in the proof of \Cref{thm:NPhardnessThreshold}.

\begin{corollary}[of \cref{thm:NPhardnessThreshold}]\label{thm:almostsurehardness}
     The almost-sure  problem is $\NP$-hard.
\end{corollary}

   




\section{Experimental results}\label{sec:experiments}


We implement two algorithms to compute the max-min value and one for the almost-sure case. The algorithm for the value problem is a global reduction that maps the game to a single ETR formula. The second algorithm uses a value iteration (VI) scheme to compute the value. 
For the almost-sure case, we implement one algorithm based on its SAT encoding. 
To establish a baseline and contextualise our runtime, we also compare our approach with the state-of-the-art PRISM-games~\cite{KN0S20}, which solves the two-player version of concurrent games. 
PRISM-games is particularly relevant as it handles both the quantitative (max-min value) and qualitative (almost-sure) versions of this problem, but where the team players have access to shared randomness. PRISM-games, therefore, solves a strict and computationally easier subcase of our setting. The source code is available at \href{https://github.com/alipashamontaseri/Team-Concurrent-Game}{https://github.com/alipashamontaseri/Team-Concurrent-Game}.

\subsection{Algorithms}

We evaluate three algorithms for the max-min value. The first is a global reduction. The second is a value iteration framework, for which we provide three distinct implementations. Each implementation provides a \emph{sound under-approxima\-tion} of the value.  
The third is a baseline based on PRISM-games~\cite{KN0S20}.

\begin{description}
    \item[Algorithm 1: ETR-Direct.] This method implements the reduction from Section~\ref{sec:ETReasiness}. It translates the entire game into a single ETR formula. It then computes the game value using binary search combined with SMT solvers.

    \item[Algorithm 2: Value iteration (VI).] ETR-Direct scales poorly on games with many states. Therefore, we use the value iteration scheme from Section~\ref{sec:OptimalityofMemorylessStrategies}. This method updates values iteratively by solving local one-shot games. In these local games, we calculate the highest probability the team can guarantee to gain in a single step, based on the current value estimates of the states. Based on this, we update the estimates of state values.
    The value iteration algorithm is a  lower bound for the max-min value of the game, and improves it every step. Therefore, the values act as a sound under-approximation.   
    We propose three implementations for solving these local instances:
    
    \begin{enumerate}
        \item \textbf{VI-ETR:} This implementation uses an SMT solver to compute each local instance reliably. This guarantees precision but may take longer on computationally difficult instances.
        
        \item \textbf{VI-OPT:} This implementation treats each one-shot game as a non-linear optimisation problem to be faster. We use Sequential Least Squares Programming (SLSQP)~\cite{kraft1988software} to solve them. Because the constraints are smooth (in $C^\infty$), the solver uses analytical gradients. This method yields an under-approximation of the value because the solver may converge to a local optimum rather than the global one. However, our experiments show that the computed value remains close to the true value.
        
        \item \textbf{VI-Hybrid:} This implementation combines SLSQP optimisation with exact solving. For each local game, we first generate a candidate value $v^*$ using SLSQP. We then query the SMT solver to verify if a value $v \ge v^* + \epsilon$ is feasible. If the query is unsatisfiable, we accept $v^*$ as the global optimum. If satisfiable (indicating $v^*$ is a local optimum), we solve the instance with binary search using $v^*$ as a lower bound. Thus, this reliably computes the value of the game.
    \end{enumerate}

    \item[Algorithm 3: Quantitative PRISM (Baseline).] We include the state-of-the-art algorithm from PRISM-games ~\cite{KN0S20}. It assumes that team players share randomness (whereas our setting assumes that team players do not share randomness). Our setting is more general, as shared randomness can be modelled by a single player.
\end{description}

The optimisation strategy used in VI-OPT and VI-Hybrid (SLSQP) is explained in detail in \cref{app:slsqp}.

For the almost-sure reachability problem, we evaluate two algorithms:

\begin{description}
    \item[Algorithm 4: SAT-Direct.] This method implements the encoding described in ~\cref{sec:almost-sure-sat}. It translates the almost-sure problem into a SAT. We then use a SAT solver to determine if an almost-surely winning collective strategy exists. This algorithm always gives a verification guarantee. 
    
    \item[Algorithm 5: Qualitative PRISM (Baseline).] We use the qualitative verification engine of PRISM-games~\cite{KN0S20} as a baseline. Similar to Algorithm 3, this tool solves the almost-sure problem under the assumption of shared randomness, which is only a comparison point for our general setting.
\end{description}









\subsubsection{Implementation.}

We implemented a prototype solver in Python 3.11.14 to evaluate these algorithms. We conducted all experiments on a machine running Ubuntu 24.04 LTS, equipped with an Intel Core Ultra 5 225U processor and 16 GB of RAM. We set a timeout of 600 seconds for each algorithm on each instance. For the max-min problem, we perform binary search until we achieve a precision of $\epsilon = 10^{-4}$.

\begin{enumerate}
    \item \textbf{SMT solver:} We use the Z3 SMT solver~\cite{Zthree} for all ETR queries. This handles Algorithm 1, and the verification steps in the value iteration variants.

    \item \textbf{SLSQP optimiser:} We use the SLSQP algorithm provided by SciPy~\cite{Scipy}. This optimiser solves the local games in VI-OPT and generates the candidate values for VI-Hybrid.

    \item \textbf{Stopping criteria:} We terminate the value iteration when the values stabilise. We stop the process when the maximum change across all states drops below $\epsilon' = 10^{-4}$ and report the final values.

    \item \textbf{SAT solver:} For the almost-sure reachability problem (Algorithm 4), we use the \texttt{minisat22} SAT solver \cite{sorensson2010minisat} provided by the PySAT library.
\end{enumerate}

\begin{figure*}[h!]
    \centering
    \includegraphics[scale=0.32,alt={A grid of six line charts plotting execution time in seconds versus the number of states on logarithmic scales. The charts are arranged in three rows corresponding to different benchmarks: Pursuit-Evasion with Rendezvous, Robot Coordination on a Grid, and Jamming Multi-Channel Radio Systems. The left column shows results for Individual Randomness, and the right column shows Shared Randomness. Lines represent the performance of ETR-Direct, VI-ETR, VI-OPT, VI-Hybrid, and PRISM-Quant algorithms. The charts demonstrate that ETR-Direct frequently hits timeouts, while VI-OPT generally scales the best for large state spaces.}]{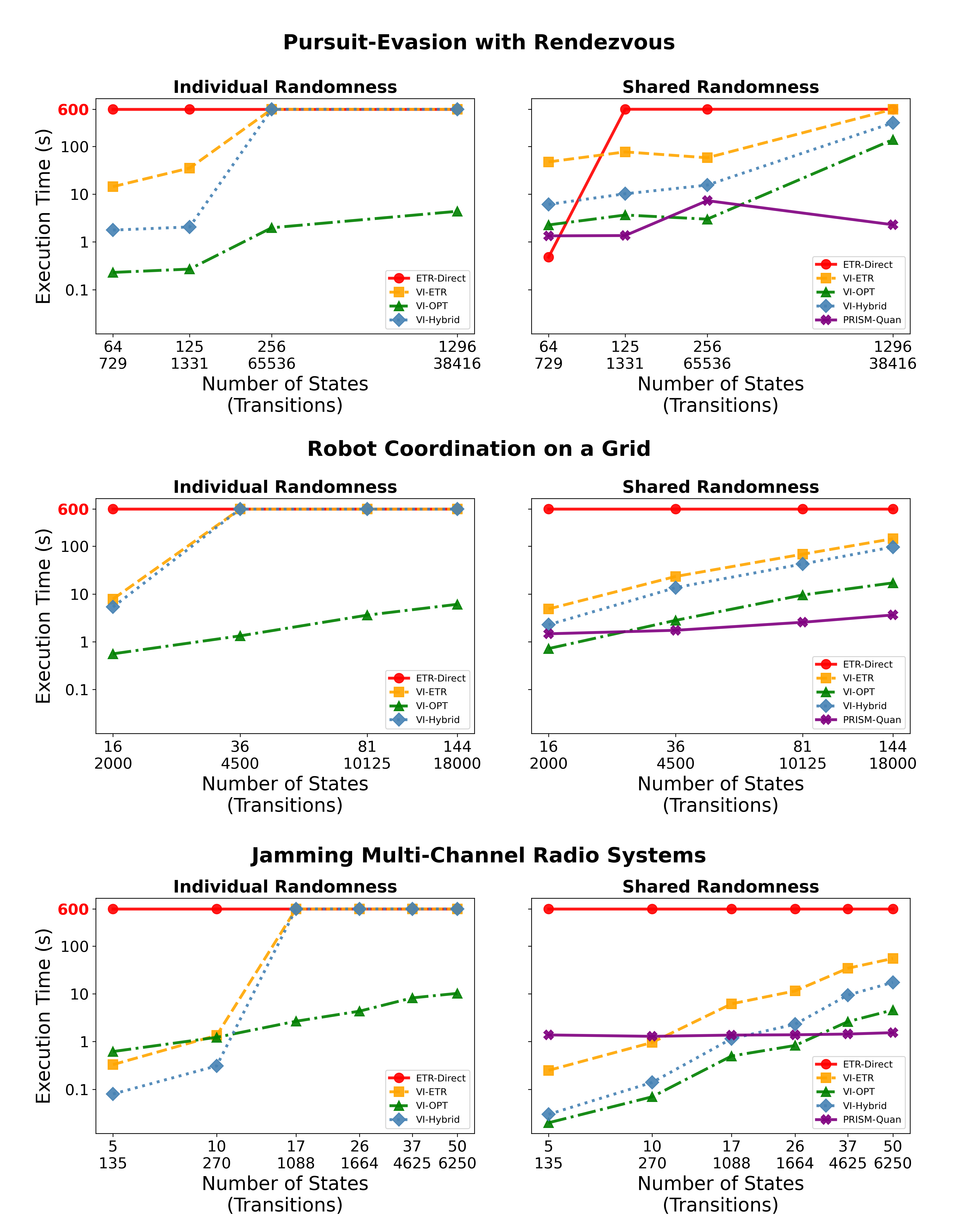}
    \caption{Execution time of the algorithms across three benchmarks. The plots report execution time in seconds vs.\ the state space size: both in log scale.}
    \label{fig:scalability}
\end{figure*}

\subsection{Benchmarks}

We evaluate our algorithms on three distinct benchmarks. The first is Pursuit-Evasion with Rendezvous, a variant of the game from \cite{parsons2006pursuit}. In this setting, a cooperative team of agents must meet at \emph{one vertex} in a directed graph while avoiding a pursuer. The second is Robot Coordination, adapted from PRISM-games \cite{kwiatkowska2021automatic}, where a team of robots navigates a grid against adversarial wind conditions to reach a target configuration. The third is Jamming Multi-Channel Radio Systems, also adapted from PRISM-games \cite{kwiatkowska2021automatic}, which models sensors transmitting packets over frequency channels subject to adversarial jamming and collisions. We provide formal definitions, dynamics, and objective functions for all three benchmarks in Appendix~\ref{app:experiments_benchmarks}.

\begin{figure*}[h!]
    \centering
    \includegraphics[scale=0.25,alt={A line chart displaying execution time in seconds versus the number of states for the Almost-sure Pursuit-Evasion with Rendezvous benchmark. The axes are logarithmic. Three lines compare the performance of SAT-Direct with Individual Randomness, SAT-Direct with Shared Randomness, and PRISM-Qual with Shared Randomness. The chart shows both SAT-Direct methods scaling efficiently, staying highly competitive with the PRISM-Qual baseline across increasing state space sizes.}]{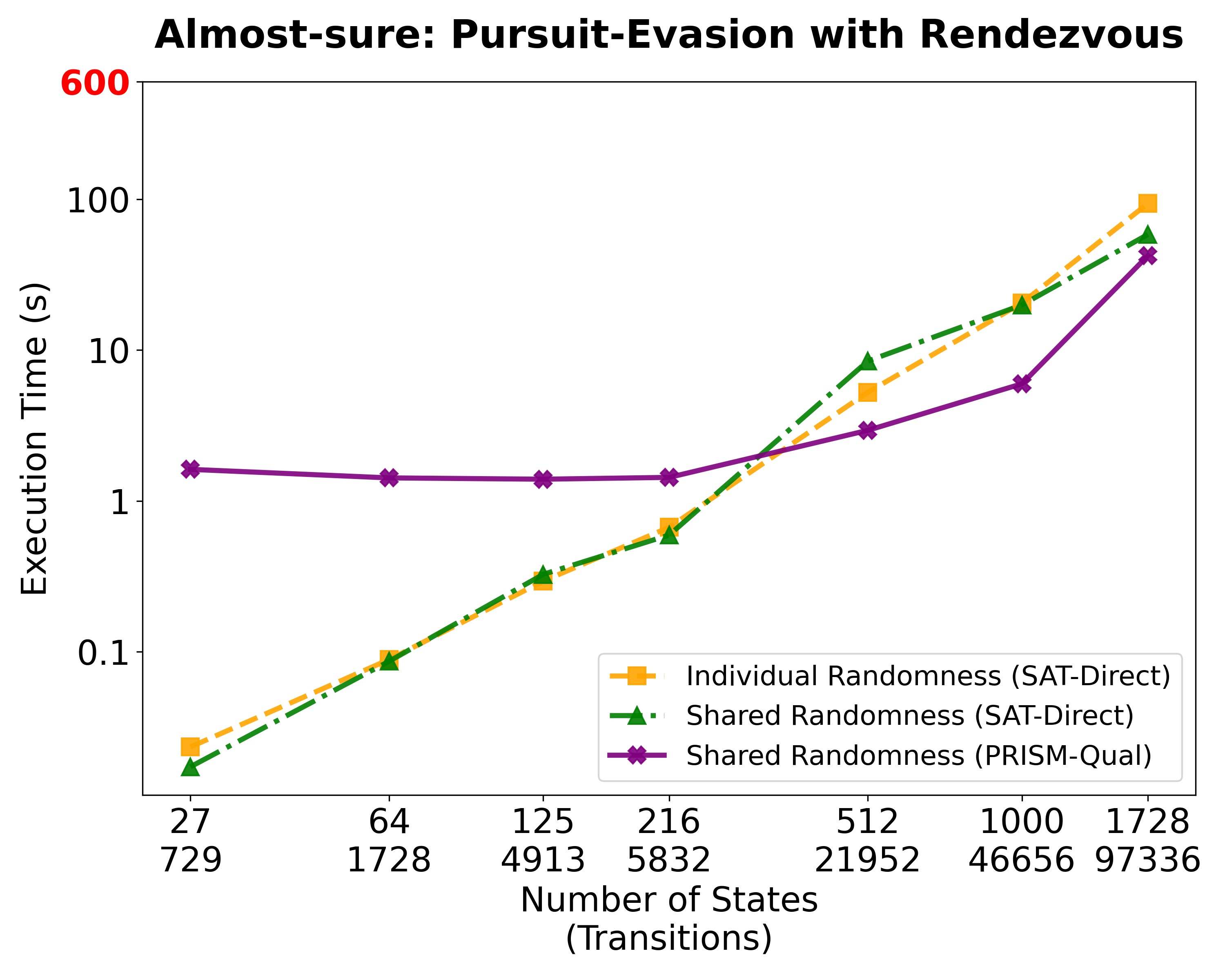}
    \caption{The plot compares the execution time (in seconds, logarithmic scale) of SAT-Direct under individual and shared randomness constraints against the baseline qualitative solver (which assumes shared randomness)}
    \label{fig:alomst-scalability}
\end{figure*}

\subsection{Results}

We evaluate our algorithms in two settings. The first is the general case of \textit{individual randomness}. The second is the special case of \textit{shared randomness}. In the latter, we model the team as a single ``meta-player'' with joint actions, which allows us to compare our results against PRISM-games. Detailed numerical data is provided in Appendix \ref{app:experiment_results}.

\subsubsection{Max-min problem.} Across both settings, \textbf{ETR-Direct} fails to solve even small instances within the timeout. In the case of individual randomness (Figure~\ref{fig:scalability}, left), the exact and hybrid VI approaches \textbf{VI-ETR}, \textbf{VI-Hybrid} solve small instances, but their execution time increases drastically as the state space grows. Conversely, the optimisation-based variant \textbf{VI-OPT} scales effectively, solving the largest instances across all benchmarks within the time limit.

In terms of precision, \textbf{VI-OPT} yields tight under-approximations in both settings, consistently converging to values very close to the exact solutions. While PRISM outperforms our general-purpose solvers in the special case of shared randomness (Figure~\ref{fig:scalability}, right), \textbf{VI-OPT} stays competitive.

\subsubsection{Almost-sure problem.}  We analyse the Pursuit-Evasion benchmark. We selected this instance because it is the only one where the team has non-trivial strategies to win with probability $1$. Figure~\ref{fig:alomst-scalability} presents the results.

First, we look at the case of \textit{shared randomness}. Here, our \textbf{SAT-Direct} algorithm performs well compared to PRISM-games. On smaller instances (Scenarios 1 to 4), our method is actually faster than PRISM because it has less overhead. In the largest instances (Scenarios 6 and 7), PRISM scales slightly better, but our solver remains competitive, despite solving a more difficult problem.

Second, we look at the case of \textit{individual randomness}. This problem is much harder because the players cannot coordinate their actions. As expected, the time needed to find a solution grows much faster than in the shared case. Despite this added difficulty, \textbf{SAT-Direct} is successful. It manages to solve large games with over 97,000 transitions within the time limit.

\section{Towards Individually Randomised ATL}\label{sec:logic}
We conclude this paper by proposing a new logic to capture our concept of individual randomisation.
A classical tool to define specifications over reactive systems is the \emph{Alternating Temporal Logic} (ATL)~\cite{AHK02}, in which one can encode the fact that a coalition of players can guarantee surely an objective that is expressible in Linear Temporal Logic (LTL).
Since only sure guarantees are considered here, the semantics of ATL need not consider randomised strategies.
This changed with the introduction of \emph{Randomised ATL} (RATL) \cite{AHK07} (and PATL~\cite{CL07}).
There, sure guarantees are extended with almost-sure and limit-sure ones, and the semantics allows players to use randomness.
For example, the formula $\team{C}{\almost} \Box q$ means that the coalition $C$ has a ``collective strategy'' to ensure that every state in the play satisfies $q$  with probability $1$.
However, the RATL semantics defines a collective strategy for the coalition $C$ as a mapping from histories to distributions over the whole set of tuples of actions available for the players in $C$, implicitly assuming that those players can randomise based on a secret common source of randomness---which can then be reduced to a $2$-player setting.
Our notion of collective strategy, where players in a coalition can design their strategy together but randomisation is an individual process, cannot be captured by this logic.
In this section, we introduce \emph{Individually Randomised ATL} (or IRATL), whose semantics bridge this gap.

\subsubsection{Syntax.}
Our logic is defined to extend the logic RATL.

\begin{definition}[IRATL formula]
    Let $\Atoms$ be a set of atoms.
    Let $\PlayerSet$ be a set of players.
    An \emph{IRATL formula} is either:
    \begin{itemize}\item a proposition $q \in Q$;
        
        \item or a formula $\neg \phi$ or $\phi \vee \psi$, where $\phi$ and $\psi$ are IRATL formulas;

        \item or a formula $\teamr{C}{w}{r} \Circle \phi$, $\teamr{C}{w}{r} \Box \phi$,   or $\teamr{C}{w}{r} \phi \until \psi$, where $C \subseteq \PlayerSet$ is a coalition of players, where $r \in \{\sh, \ind\}$ is a \emph{randomisation type}, and where $w \in \{\sure, \almost, \limit\} \cup \{> t, \geq t \mid t \in [0,1)\}$ is a \emph{winning condition}, and where $\phi$ and $\psi$ are IRATL formulas.
    \end{itemize}
\end{definition}

The operators $\teamr{}{w}{r}$ are called \emph{path quantifiers}, the operators $\Circle$, $\Box$, and $\until$ are called \emph{temporal operators}, and subformulas of the form $\Circle \phi$, $\Box \phi$ or $\phi \until \psi$ are called \emph{path formulas}.

\subsubsection{Semantics.}

    Given a game structure $G$, a set of atomic propositions $Q$, and a labelling function $L: S \mapsto 2^Q$, we define the semantics of IRATL by cross-induction as follows.
    Given a state $s \in S$, we define:
    \begin{itemize}
        \item $s \models q$ if we have $q \in L(s)$;
        
        \item $s \models \neg \phi$ if we do not have $s \models \phi$, and $s \models \phi \vee \psi$ if we have $s \models \phi$ or $s \models \psi$;
        
        \item $s \models \teami{C}{\sure} \theta$ if in the game structure $G$, from the state $s$, there exists a strategy profile $\bstrat_C$ such that for every strategy profile $\bstrat_{\PlayerSet \setminus C}$ and every play $\pi$ compatible with the strategy profile $\bstrat$, we have $\pi \models \theta$;
        
        \item $s \models \teami{C}{\almost} \theta$ if in the game structure $G$, from the state $s$, there exists a strategy profile $\bstrat_C$ such that for every strategy profile $\bstrat_{\PlayerSet \setminus C}$, we have $\prob_{\bstrat}\{\pi \in S^\omega \mid \pi \models \theta\} = 1$;
        
        \item $s \models \teami{C}{\limit} \theta$ if in the game structure $G$, from the state $s$, for every $\epsilon > 0$, there exists a strategy profile $\bstrat_C$ such that for every strategy profile $\bstrat_{\PlayerSet \setminus C}$, we have $\prob_{\bstrat}\{\pi \in S^\omega \mid \pi \models \theta\} > 1-\epsilon$;
        
        \item $s \models \teami{C}{>t} \theta$ (resp. $\teami{C}{\geq t} \theta$) if in the game structure $G$, from the state $s$, there exists a strategy profile $\bstrat_C$ such that for every strategy profile $\bstrat_{\PlayerSet \setminus C}$, we have $\prob_{\bstrat}\{\pi \in S^\omega \mid \pi \models \theta\} > t$ (resp. $\geq t$);
        
        \item $s \models \teamr{C}{w}{\sh} \theta$ is defined analogously, but in the game structure $G_C$, in which all $C$ players have been merged into one player who can pick an action profile at each state.
    \end{itemize}
    
    On the other hand, given a play $\pi = \pi_0 \pi_1 \dots$, we define:
    
    \begin{itemize}
        \item $\pi \models \Circle \phi$ if we have $\pi_1 \pi_2 \dots \models \phi$;
        
        \item $\pi \models \Box \phi$ if for every $n \in \Nb$, we have $\pi_n \models \phi$;
        
        \item $\pi \models \phi \until \psi$ if there exists $n \in \Nb$ such that we have $\pi_n \models \psi$, and $\pi_k \models \phi$ for every $k < n$.
    \end{itemize}

Other classical symbols can be derived: for example, the formula $\top$ is short for $q \vee \neg q$ for some $q$, and the formula $\Diamond \phi$ for $\top \until \phi$.

\subsubsection{Model-checking.}

The algorithms presented in the previous sections enable us to define a decidable fragment of IRATL.

\begin{theorem}
    Given a game structure $G$, a state $s$ and a formula $\phi$ that does not contain the symbols $\teami{C}{\limit}$, $\teami{C}{\geq t}$, or $\teami{C}{>t}\Box$, deciding whether in the game $G$ we have $s \models \phi$ can be done in $\PSpace$, and in polynomial time if the size of the game structure is fixed.
\end{theorem}

\begin{proof}
    As a preliminary remark, let us remind that when a coalition $C$ seeks to achieve an objective $\theta$ against all strategies of the other players, those can be modelled by a single opponent, as they do not need any form of randomisation to respond to the coalition's strategy.
    We can then apply the following recursive algorithm.
    Given an IRATL formula~$\phi$:
    \begin{itemize}
        \item if $\phi = q, \neg \psi$,  or $\psi \vee \chi$, treat it recursively as expected.

        \item If $\phi = \teamr{C}{w}{\sh} \theta$, then apply the same recursive operations as for RATL~\cite{AHK07} and PATL~\cite{CL07}.

        \item If $\phi = \teami{C}{\sure} \theta$: since the coalition must surely reach the target, randomising is useless, and the coalition can then be seen as a single player.
        We can therefore apply the algorithms cited at the previous point.

        \item If $\phi = \teami{C}{\almost} \Box \psi$: this formula is semantically equivalent to $\teami{C}{\sure} \Box \psi$ (see Appendix \ref{app:logic} for a proof).
        We can therefore apply the previous point.
        
        \item If $\phi = \teami{C}{\almost} \Circle \psi$ (resp. $\teami{C}{>t} \Circle \psi$), first compute the set $T$ of states that satisfy the formula $\psi$.
        Then, consider the game in which the team $C$ wants to reach the set $T$, and all states except $s$ are made absorbing.
        Then, decide whether the team can guarantee reaching the set $T$ almost surely (resp. with probability greater than $t$), using \Cref{thm:thmETRComplete} (resp. \Cref{thm:almostsureeasiness}).

        \item If $\phi = \teami{C}{\almost} \psi \until \chi$ (resp. $\teami{C}{>t} \psi \until \chi$), first compute the sets $U$ and $T$ of states that satisfy the formulas $\psi$ and $\chi$, respectively.
        Consider the game in which all states $s' \in S \setminus U$ are made absorbing.
        Then, decide whether in this game, the team $C$ can guarantee reaching the set $T$ almost surely (resp. with probability greater than $t$), using \Cref{thm:thmETRComplete} (resp. \Cref{thm:almostsureeasiness}).
    \end{itemize}

    The size of the recursion stack is bounded by the size of the formula, and each recursion consists in a polynomial-space algorithm, hence the conclusion.
\qed\end{proof}

The decidability of the full logic IRATL will be obtained if one can handle the two following problems, which we leave as open questions.

\begin{problem}[Limit-sure problem]\label{problem:limit}
    Given a game $\Game$, does it hold that for every $\epsilon > 0$, there exists a collective strategy $\bstrat_\Players$ such that for any strategy $\strat_\opp$, the probability of reaching the set $\Terminals$ under the strategy profile $\bstrat$ is greater than~$1-\epsilon$?
\end{problem}

\begin{problem}[Safety problems]\label{problem:safety}
    Given a game $\Game$, can the team guarantee (with probability greater than a given threshold, or limit-surely) that the set $T$ is \emph{not} eventually reached?
\end{problem}

The reader may already be convinced that the limit-sure problem requires new techniques, but believe that the safety threshold problem can be treated in a way analogous to the reachability threshold problem.
This is not the case, because \Cref{thm:memmless-optimal} does not extend to safety games: as shown in the work of de Alfaro et al.~\cite[Theorem~8]{AHK07}. Even in a two-player setting, infinite memory may be required to guarantee that a safety objective is satisfied with positive probability.
However, in that setting, this problem can be decided by taking the perspective of the opponent---by asking whether the opponent has a strategy to make the player lose with probability at least $1-t$. Such arguments no longer work in our setting, since the max-min value does not equal the min-max value.
Hence, the safety threshold problem seems to be more challenging than its reachability counterpart. 

\section{Discussion}
\label{sec:discussion}
We introduced the logic IRATL, that extends the well-studied ATL to support cases where players randomise independently without a common coin. We showed that this distinction is not just semantic but requires new algorithms.

IRATL describes games where players form teams, but must randomise individually.
The exact complexity of solving these games is left as an open problem. We show $\exists\Rb$-easiness and $\NP$-hardness for the threshold problem, which is also $\SQRTSUM$-hard, 
following from the two-player case~\cite{EY15}. It would be interesting to know the exact complexity of this problem.

Second, from a practical perspective, our current value iteration method relies on asymptotic convergence without a certified stopping criterion. Even in the simpler two-player concurrent setting, it is known that a double-exponential number of iterations might be required to approximate the value to within an $\epsilon$-approximate value~\cite{HIM11}. A promising avenue for future work is to extend recent results on stopping criteria for concurrent games~\cite{KMW23}. By developing techniques to approximate the value from both above and below, we could rigorously guarantee $\epsilon$-optimality upon termination. We leave as future work the integration into established model checkers such as PRISM-games~\cite{KN0S20} or STORM~\cite{STORM22} enabling model checking tools that accurately model decentralised systems.

 \begin{credits}
 \subsubsection{\ackname} This work is a part of project VAMOS that has received funding from the European Research Council (ERC), grant agreement No 101020093.
 Part of this work was realised when the first author was an FNRS aspirant at Université libre de Bruxelles.
 \end{credits}
%
%
%

\bibliographystyle{splncs04}
\bibliography{biblio}

\appendix

\section{The existential theory of the reals}\label{app:additionalDefinitions}
\begin{definition}[Existential theory of the reals]
    A \emph{formula of the existential theory of the reals}, or \emph{ETR} for short, is a formula of the form 
    $$\exists x_1 \dots \exists x_k~\phi(x_1, \dots, x_k)$$ where $\phi$ is a quantificator-free formula written with the symbols $0$, $1$, $=$, $\leq$, $<$, $+$, $-$, $\times$, $\wedge$, $\Leftrightarrow$, $\neg$, and parentheses, with the expected syntactic rules and semantics.
\end{definition}

We write $\exists\Rb$ for the complexity class of problems that can be reduced in polynomial time to the problem of deciding the validity of an ETR formula. 
That class is known to be included in $\PSpace$~\cite{Can88}. It is also easy to see that $\exists\Rb$ contains the class $\NP$.

\section{Appendix for~\cref{sec:Threshold}}\label{app:threshold}
\subsection{A first lemma}
We prove a simple lemma that we use in the proof of \cref{thm:thmETRComplete}.
\begin{restatable}{lemma}{thresholdsupattain}{\label{lemm:sup-attain}}
    For any valuation $x$ and any state $s$, the supremum in the definition of $\mathsf{Pre}(x)(s)$ is attained.
    Formally, there exists a memoryless collective strategy $\bar{\sigma}_{\Players}^* \in \Lambda_{\Players}$ 
    such that:
    $$
        \mathsf{Pre}(x)(s) = \inf_{\sigma_\opp \in \Lambda_\opp} \mathsf{Pre}_{\bar{\sigma}_{\Pi}^*, \sigma_\opp}(x)(s)
    $$
\end{restatable}

\begin{proof}
   First, the set $\Lambda_{\Pi}$ of selectors for the team 
    is defined by the product of probability distributions 
    over finite action sets for each player. 
    In mathematical terms, this set is compact 
    (it is closed and bounded).

    Second, consider the function $f$
    that maps a team selector to its guaranteed value:
    $$
        f: \bstrat_{\Players} \mapsto \inf_{\strat_{\opp}} \mathsf{Pre}_{\bstrat}(x)(s)
    $$
    The calculation of $\mathsf{Pre}$ involves adding and multiplying the probabilities defined in 
    $\bstrat_{\Players}$ and $\strat_{\opp}$. 
    Because of this, the function $f$ is continuous.

    Finally, we apply the extreme value theorem \cite{WR76}. 
    This theorem states that a continuous function defined on a compact set must achieve a maximum value. 
    Therefore, there is a specific strategy 
    $\bstrat_{\Players}^*$ that maximises the value.
\qed\end{proof}

\subsection{Proof of \Cref{lm:VIconvergencetovalue}} \label{app:VIconvergencetovalue}

\lmVIconvergencetovalue*

\begin{proof}
We first prove that the value iteration sequence converges.

    \begin{restatable}{proposition}{thresholdconvergence}{\label{lem:convergence}}
    There exists $w = \lim_{k \to \infty} u_k$.
\end{restatable}

\begin{proof}
    First, we show that the predecessor operator 
    $\mathsf{Pre}$ is monotone. 
    Let $x$ and $y$ be two valuations such that 
    $x \le y$ (i.e., $x(s) \le y(s)$ for all $s \in \Vertices$).
    Consider any fixed collective selector $\bxi_\Players$ 
   and any selector $\xi_\opp$ for the opponent.
   Given a state $s$, the corresponding expected payoff is:
   $$
    \mathsf{Pre}_{\bxi}(x)(s) 
    \;=\; 
    \sum_{\baction \in \Av_{\Players}(s)} \sum_{b \in \Av_{\opp}(s)} \sum_{t \in \Vertices} 
    \bxi_{\Players}(s)(\baction) \cdot \xi_{\opp}(s)(b) \cdot \delta(s, \baction, b)(t) \cdot x(t).
    $$

    This expression is a weighted average of the values $x(t)$ with non-negative weights 
    (i.e., $\bxi_\Players, \xi_{\opp}, \delta$ have non-negative values).
    Therefore, if $x(t) \le y(t)$ for all states $t$, the weighted average for $x$ cannot exceed the weighted average for $y$. Thus:
    \[
        \mathsf{Pre}_{\bxi}(x)(s) \;\le\; \mathsf{Pre}_{\bxi}(y)(s).
    \]
    Since taking the infimum (over the opponent strategies)
    and the supremum (over team strategies) preserves the order, we have:
    \[
        \mathsf{Pre}(x) \le \mathsf{Pre}(y).
    \]

    Now we consider the sequence $(u_n)_{n\ge 0}$ and show the convergence in three steps.
    
    \paragraph{Step 1: $u_0 \le u_1$.}
    If $s \in T$, then $\mathsf{Pre}(u_0)(s) = 1 = u_0(s)$ because the set $T$ is assumed to be absorbing. 
    If $s \notin T$, then $u_0(s) = 0 \le \mathsf{Pre}(u_0)(s)$, where the inequality holds because $\mathsf{Pre}$ is defined by
    non-negative probabilities.
    
    \paragraph{Step 2: Monotonicity of $(u_n)_n$.}
    Assume $u_n \le u_{n+1}$ for some $n \ge 0$. Since the operator $\mathsf{Pre}$ is monotone, applying it to both sides gives:
    \[
        u_{n+1} = \mathsf{Pre}(u_n) \;\le\; \mathsf{Pre}(u_{n+1}) = u_{n+2}.
    \]
    By induction, the sequence is non-decreasing: $u_0 \le u_1 \le u_2 \le \dots$.
    
    \paragraph{Step 3: Boundedness.}
    All valuations in the sequence are in $[0,1]^{\Vertices}$.  
    Thus for each $s\in \Vertices$, the real sequence $(u_n(s))_{n\ge 0}$ is
    non-decreasing and bounded above by $1$, hence it converges.
\qed\end{proof}

 Now, 
 define $\prob(\Diamond \Terminals)(s)$ as the maximum probability 
 with which the team \emph{can} guarantee reaching $T$ starting from $s$:

 $$
 \prob(\Diamond \Terminals)(s) = \sup_{\bstrat_{\Players}}
   \inf_{\strat_{\opp}}
   \prob_{\bstrat}(\Diamond \Terminals)(s)
 $$
 Where $\prob_{\bstrat}(\Diamond T)(s)$
 denotes the probability of reaching $T$, starting from $s$, 
 if the team plays with collective strategy $\bstrat_{\Players}$, 
 and the opponent plays with strategy~$\strat_{\opp}$.
 By the same approach, we define 
 $\prob(\Diamond^n T)(s)$ as the maximum probability that the team can guarantee of reaching $\Terminals$ within at most $n$ steps, and
 $\prob_{\bstrat}(\Diamond^n T)(s)$
 as the same probability under the strategy profile
 $\bstrat = (\bstrat_{\Players},\strat_{\opp})$. 
 
 We want to prove that for all $s$, we have 
 $w(s) = \prob(\Diamond \Terminals)(s)$.
 This will be a consequence of the two following propositions.

\begin{proposition}\label{lemm:w-greater-value}
    For all $\epsilon > 0$, the team players have a collective
    strategy $\bstrat_{\Players}$ that guarantees reaching $T$
    from every state $s$ with probability $w(s) - \epsilon$.
\end{proposition}

\begin{proof}
    For each $n\ge 0$, 
    We define $\bxi^n_{\Players}(s)$ as the collective selector that realises the supremum in the definition of $u_n(s)$. 
    Based on Lemma~\ref{lemm:sup-attain}, 
    such a strategy exists.
    Then, we have:
    \[
    u_n(s)
    =
    \inf_{\xi_\opp}
    \mathsf{Pre}_{\bar{\xi}^n_{\Players},\xi_
    \opp}(u_{\,n-1})(s)
    \qquad\text{for all }s.
    \]
    
    For a fixed $n\ge 0$ we define a collective strategy $\bar{\tau}^n_{\Players}$ 
    that plays
    $\bar{\xi}^n_{\Players}$ in the first step, 
    then 
    $\bar{\xi}^{\,n-1}_{\Players}$ 
    in the second step and so on.
    More formally, for all $1 \leq i \leq n$, in the $i$-th step, 
    the team plays according to the $\bar{\xi}^{n-i+1}_{\Players}$.
    In the following steps, it behaves arbitrarily.     

    We show by induction on $n$ that for every state $s$ and all strategies
    $\xi_\opp$ of the opponent, we have:
    \[
    \prob_{\bar{\tau}^n_{\Players},\xi_\opp}(\Diamond^{n}\Terminals)(s) \ge u_n(s).
    \]
    For $n=0$ the claim is correct, since $\Diamond^{0}T$ simply means that
    $s\in T$, 
    and by definition we know that $u_0(s)=1$ 
    if and only if $s \in \Terminals$.  
    Assume the claim holds for $n-1$.  
    Starting from a state $s$, the first step is played using $\bar{\xi}^n_{\Players}$.
    Conditioning on the successor $t$, if $t\in T$ the objective is already
    satisfied, and if $t\notin T$ then by the induction hypothesis the probability
    of reaching $T$ within the next $n-1$ steps is at least $u_{n-1}(t)$.
    Therefore, we have:
    \[
    \prob_{\bar{\tau}^n_{\Players},\xi_\opp}(\Diamond^{n}\Terminals)(s)
    \;\ge\;
    \mathsf{Pre}_{\bar{\xi}^n_{\Players},\xi_\opp}(u_{\,n-1})(s).
    \]
    Taking the infimum over all opponent strategies and using the definition of
    $\bar{\xi}^n_{\Players}$, we obtain:
    \[
    \forall \xi_\opp,
    \prob_{\bar{\tau}^n_{\Players},\xi_\opp}(\Diamond^{n}\Terminals)(s)
    \;\ge\;
    u_n(s).
    \]
    This completes the induction.
    
    Since $u_n\to w$ and $S$ is finite, there exists $n$ such that
    $|u_n(s)-w(s)|<\varepsilon$ for all $s$.  
    For this $n$, the strategy $\bar{\tau}^n_{\Players}$ satisfies, for all $s$ and all
    $\sigma_\opp$ we have:
    \[
    \prob_{\bar{\tau}^n_{\Players},\sigma_\opp}(\Diamond^{n}\Terminals)(s)
    \;\ge\;
    u_n(s)
    \;\ge\;
    w(s)-\varepsilon.
    \]
    This proves the proposition.
\qed\end{proof}

\begin{proposition}\label{lemm:w-less-value}
    For every state $s$,
    team players do not have a collective strategy
    $\bstrat_{\Players}$ that guarantees reaching $T$
    with probability greater than $w(s)$.
\end{proposition}

\begin{proof}
    We show that for any collective strategy 
    $\bstrat_{\Players}$ chosen by the team, 
    the opponent has a counter-strategy $\strat_{\opp}$ 
    that guarantees the probability of reaching the target $T$ never exceeds $w(s)$. 
    We prove this by showing that the opponent can keep the
    probability of reaching $\Terminals$ within any number of turns $n$ 
    ($\Diamond^n T$ objective) below the valuation $u_n$.

    We use induction on $n$. For $n=0$, reaching the target in zero steps is only possible if we are already there, so the probability is $1$ for $s \in T$ and $0$ otherwise.
    This matches the definition of $u_0$ since $u_0$ assigns $1$ only to state $s \in \Terminals$.

    Now, assume that the claim is true for some $n$.
Let $s$ be a state, and let $\bstrat_\Players$ be a collective strategy.
By definition of the predecessor operator $\Pre$, we have:
$$\inf_{\xi_{\opp} \in \Lambda_{\opp}}
        \Pre_{\bxi}(u_n)(s) \leq u_{n+1}(s)$$
Moreover, there is a selector $\xi_\opp$ which realises this infimum, and therefore such that $\Pre_{\bxi}(u_n)(s) \leq u_{n+1}(s)$.
Now, let us consider the following strategy for the opponent.
First, she plays the selector $\xi_\opp$.
Then, let $t$ be the state that is reached: from $t$, by induction hypothesis, the opponent can respond to the team's strategy to ensure that the probability of reaching the set $T$ during the next $n$ steps does not exceed $u_n(t)$.
The opponent then plays that strategy.
Then, from $s$, the probability of reaching $T$ is at most $u_{n+1}(s)$.

    As $n$ grows, the values $u_n$ converge to the limit $w$. 
    Since the opponent can keep the team's success after $n$ steps below $u_n$
    for every specific $n$, the opponent can also keep the total probability of eventually reaching $T$ from exceeding the limit $w(s)$. This shows that no collective strategy can guarantee a payoff strictly greater than $w(s)$.

    More formally:

    \begin{align*}
        &\forall n, \forall {\bstrat_{\Players}}, \exists {\strat_{\opp}}, 
        \prob_{\bstrat} (\Diamond^n T) \leq u_n 
        && \text{(proven by induction)}
        \\[1em]
        \implies \quad 
        &\forall {\bstrat_{\Players}}, \exists {\strat_{\opp}}, \forall n,
        \prob_{\bstrat} (\Diamond^n T) \leq w 
        && \text{(since $u_n \le w$)}
        \\[1em]
        \implies \quad 
        &\forall {\bstrat_{\Players}}, \exists {\strat_{\opp}}, \sup_n 
        \prob_{\bstrat} (\Diamond^n T) \leq w 
        && \text{(supremum over $n$)}
        \\[1em]
        \implies \quad 
        &\forall {\bstrat_{\Players}}, \exists {\strat_{\opp}}, 
        \prob_{\bstrat} (\Diamond T) \leq w
        && \text{(monotonicity of reach probability)}
        \\[1em]
        \implies \quad 
        &\sup_{\bstrat_{\Players}} \inf_{\strat_{\opp}} 
        \prob_{\bstrat} (\Diamond T) \leq w
        && \text{(properties of sup/inf)}
    \end{align*}

    The last sentence shows that for any collective strategy $\bstrat_{\Players}$ and any state $s$, 
    the opponent has a counter strategy that prevents the team from reaching
    $T$ with probability greater than $w(s)$.
\qed\end{proof}

This proves our lemma.\end{proof}

\subsection{Proof of \Cref{thm:memmless-optimal}} \label{app:memmless-optimal}

\thmmemlessoptimal*


\begin{proof}
We use the same approach as in the work of Chatterjee, de Alfaro and Henzinger~\cite{CDH12} for two-player concurrent game, where they show that memoryless strategies suffice for two-player zero-sum concurrent games.

\begin{proposition} \label{lemm:rank-definition}
    Let $n$ be the index such that the valuation $u_n$ satisfies 
    $\|w - u_n\|_\infty < \epsilon/2$ 
    (which exists since $(u_n)_n$ converges to $w$).
    For each non-terminal state $s$, 
    let the rank $\ell(s)$ be the smallest 
    $k \in \{1, \dots, n\}$ such that $u_k(s) = u_n(s)$. 
    There exists a collective selector $\bxi^s_\Pi$
    such that for any opponent action $b \in A_{\opp}$, 
    the expected value of the valuation $u_{\ell(s)-1}$ at the subsequent state $t$ satisfies:
    \[
        \sum_{\baction_\Players \in \Av_\Players(s)}
          \left(
             \bxi_\Players^s(\baction_\Players)
             \cdot \sum_{t \in \Vertices} \delta(s, \baction_\Players, b)(t) \cdot u_{\ell(s)-1}(t)
          \right)
        \geq u_{\ell(s)}(s).
    \]
\end{proposition}

\begin{proof}
    By the construction of the value iteration sequence,
    we have 
    $u_{\ell(s)}(s) = \mathsf{Pre}(u_{\ell(s)-1})(s)$. 

    We recall that $\mathsf{Pre}$ operator is defined by a supremum
    over team memoryless strategies
    and an infimum over opponent memoryless strategies:
    \[
        \mathsf{Pre}(u_{\ell(s)})(s) = 
        \sup_{\bxi_{\Players} \in \Lambda_{\Players}}
        \inf_{\xi_{\opp} \in \Lambda_{\opp}}
        \mathsf{Pre}_{\bxi}(u_{\ell(s)-1})(s)
    \]
    
    From \cref{lemm:sup-attain}, 
    there must exist a memoryless collective strategy $\bxi^s_\Pi$ that secures the value $u_{\ell(s)}(s)$ relative to the previous valuation $u_{\ell(s)-1}$. 
    This choice of action ensures that the team maintains its guaranteed success probability while accounting for the step-count progress required to reach $\Terminals$.
\qed\end{proof}

\begin{lemma} \label{lemm:memoryless-epsilon}
    There exists a memoryless collective strategy $\bstrat$ for the team that guarantees reaching $T$ from state $s$ with probability at least $w(s) - \epsilon$ for every $\epsilon > 0$.
\end{lemma}

\begin{proof}
    Consider the memoryless strategy $\bstrat$ where the team plays $\bxi^s_\Pi$ at each state $s$ as defined in \cref{lemm:rank-definition}. 

    Let $S_k$ be the random variable denoting the state at step $k$ of a play. We define the sequence of valuations along the play as $M_k = u_n(S_k)$.
    From \cref{lemm:rank-definition}, for any opponent action $b \in A_\opp$, the expected valuation in the next state is bounded by the following:
    \[
        \mathbb{E}[u_{\ell(S_k)-1}(S_{k+1}) \mid S_k = s] \geq u_{\ell(s)}(s)
    \]
    By the definition of rank $\ell$, we have $u_{\ell(s)}(s) = u_n(s)$. Furthermore, by \cref{lm:VIconvergencetovalue}, the sequence of value iterations is non-decreasing, so $u_{\ell(S_k)-1}(t) \leq u_n(t)$ for all states $t$. Substituting these into the inequality yields:
    \[
        \mathbb{E}[u_n(S_{k+1}) \mid S_k = s] \geq u_n(s)
    \]
    Because $u_n$ is bounded by \cref{lm:VIconvergencetovalue}, the sequence $M_k = u_n(S_k)$ is a bounded submartingale. By the Martingale Convergence Theorem, $M_k$ converges almost surely to a limit $M_\infty$. 

    Because the state space is finite, any play almost surely gets trapped in a set of states that share the same limit valuation $v = u_n(S_k)$. 
    Suppose that the play remains in this level set where $u_n(t) = u_n(s) = v$. From our earlier inequality, we know $\mathbb{E}[u_{\ell(s)-1}(S_{k+1}) \mid S_k = s] \geq v$. Since $u_{\ell(s)-1}(t) \leq u_n(t) = v$ for all $t$ in this level set, the only way the expectation can be at least $v$ is if $u_{\ell(s)-1}(t) = v$ for all next states $t$ that have positive probability. 
    
    By the definition of rank, if $u_{\ell(s)-1}(t) = u_n(t)$, then the rank of $t$ must be smaller than the rank of $s$. This means that, conditional on staying in the same valuation level, the rank $\ell$ strictly decreases at every step. 

    Since the rank is a natural number, it cannot decrease forever. Therefore, the play cannot remain indefinitely in a non-target level set. It must almost surely either reach a strictly higher valuation level or reach the target set $\Terminals$. 
    
    Thus, the limit $M_\infty$ is either $1$ (the play reaches $\Terminals$) or $0$. The expected value of $M_\infty$ is the probability of reaching $\Terminals$. 
    By the definition of a submartingale, the expected value of the limit is at least the initial value:
    \[
        \prob_{\bstrat}(\Diamond \Terminals)(s) = \mathbb{E}[M_\infty] \geq M_0 = u_n(s)
    \]
    By \cref{lm:VIconvergencetovalue}, there exists $n$ such that $u_n(s) \geq w(s) - \epsilon$. Therefore, the memoryless collective strategy $\bstrat$ satisfies $\prob_{\bstrat}(\Diamond \Terminals)(s) \geq w(s) - \epsilon$.
\qed\end{proof}

By \cref{lemm:memoryless-epsilon} it is true that
for every $\epsilon$, there exists a memoryless strategy 
for team players, which guarantees value $w(s) -\epsilon$. 
This shows the correctness of 
\cref{thm:memmless-optimal}.
\qed\end{proof}

\subsection{Appendix for \cref{thm:memlessOpp}}\label{app:memlessOpp}
\memlessOpp*
\begin{proof}
    We show this on a three-player example.
    
    There are three players: $\opp, P_1, P_2$.
    For the player $P_1$, 
    we define the set of actions $A_{P_1} =\{a, b, wait\}$. 
    For $\opp$ and $P_2$ the action sets are $A_\opp = A_{P_2} = \{a,b\}$.
    There are two stages in the game:
    
    \textbf{(1) Waiting stage:} This stage is active while $P_1$ chooses action \emph{wait}. 
    Based on the action of $\opp$, the game moves to $S_a$ or $S_b$.
    On $S_a$ and $S_b$ for every set of actions, the game returns to $S$.
    
    \textbf{(2) Payoff stage:} This stage occurs if $P_1$
    does not choose an action \emph{wait}. 
    The game then terminates. If both $P_1, P_2$ choose the same action
    as $\opp$ the game finishes at $\top$ (the target state), otherwise at $\perp$.

    \begin{center}
    \begin{tikzpicture}[>=latex, node distance=1.3cm, auto,
      state/.style={draw, circle, minimum width=1.2cm, font=\small},
      terminal/.style={draw, rectangle, minimum width=1.2cm, minimum height=0.8cm, font=\small},
      lab/.style={font=\scriptsize}]
    
      \node[state] (s) {$S$};
      \node[state, left=of s] (sa) {$S_a$};
      \node[state, right=of s] (sb) {$S_b$};
    
      \node[terminal, below=of s, xshift=-1.5cm] (t) {$\top$};
      \node[terminal, below=of s, xshift=1.5cm] (trap) {$\perp$};
    
      \draw[->] (s) to[bend right=20] node[lab, swap] {$(\textit{wait},*), a$} (sa);
      \draw[->] (s) to[bend left=20] node[lab] {$(\textit{wait},*), b$} (sb);
      
      \draw[->] (sa) to[bend right=20] node[lab, above] {$(*, *), *$} (s);
      \draw[->] (sb) to[bend left=20] node[lab, above] {$(*, *), *$} (s);
    
      \draw[->] (s) -- node[lab, left, pos=0.4] {$(a,a), a \text{ or } (b,b), b$} (t);
    
      \draw[->] (s) -- node[lab, right, pos=0.4] {$\text{otherwise}$} (trap);
    
    \end{tikzpicture}
    \end{center}

    During the Waiting Stage,
    the game remains in the loop
    $S \rightarrow \{S_a,S_b\} \rightarrow S$. 
    At state $S$, the next state depends only on the opponent’s action.  
    From $S_a$ and $S_b$, the game always returns to $S$. 
    During the Payoff Stage,
    the game terminates and the reward is determined.
    
    Assume the opponent is forced to play memoryless strategies and chooses action $a$ with probability $p_{\opp,a}$,
    and action $b$ with probability $1-p_{\opp,a}$ at state $S$. For computing the max-min value, the opponent fixes $p_{\opp,a}$ after knowing the strategies of the team players.

    \textbf{Case 1: The team players all play memoryless strategies.}

    If $P_1$ chooses action \emph{wait} with probability $1$, 
    the game never leaves the waiting loop, 
    so the probability of reaching the target will be $0$.
    
    If $P_1$ chooses action $\textit{wait}$ with probability
    $\alpha<1$, then at visit of the play to the state $S$, the
    probability of entering the payoff stage is  
    $1-\alpha$.
    As the following limit value is equal to 1:
    \[
    \prob(\text{ever enter payoff stage})
    \;=\;
    \lim_{k \to \infty} \bigl(1-\alpha^k\bigr)
    \;=\; 1,
    \]
    so the game reaches the payoff stage with   probability $1$, \emph{regardless of the value of $\alpha<1$}.
    
    At the payoff stage, suppose that each $P_1$ and $P_2$
    choose action $a$ with probability $p_{P_1,a},\ p_{P_2,a}$ respectively.
    
    As $P_1$ chooses action $\textit{wait}$ with probability $\alpha$, we normalise $p_{P_1,a}$ to have the actual probability of playing action $a$ at the payoff stage. Thus, we divide it by $1-\alpha$ (the probability of playing action $a$ or $b$ at the payoff stage). 
    
    As a counter strategy, 
    the opponent will choose $p_{\opp,a}$ 
    to minimise the team's chance of winning. 
    If the team is more likely to match $a$ than $b$, 
    $\opp$ will always play $b$ (and vice versa). 
    To maximise their worst-case guarantee, the team plays a strategy such that the probability of matching $a$ 
    is equal to the probability of matching $b$. 
    This requires 
    $p_{P_1,a} \cdot p_{P_2,a} = (1-p_{P_1,a}) \cdot (1-p_{P_2,a})$. 
    Thus, the best guarantee the team can have under this condition is found when they play uniformly ($p_{P_1,a}=p_{P_2,a}=1/2$). In this case, the probability of reaching $\top$ is    
    $$p_{\opp,a}\left(\frac{1}{4}\right) + (1-p_{\opp,a})\left(\frac{1}{4}\right) = \frac{1}{4}.$$

    The team players can therefore guarantee reaching $\top$ with probability $1/4$, no matter what does $\opp$ do. 
    If the team deviates from playing uniformly,
    $\opp$ can switch $p_{\opp,a}$ to make the probability even lower.
    Thus, under memoryless strategies, the team cannot guarantee a value strictly greater than $1/4$.

     \textbf{Case 2: The team players play arbitrary strategies}

     If team players use memory, 
     they can coordinate their actions based on the history of visited states. 
     Consider the following collective strategy $\bstrat_{\Players}$.
     \begin{itemize}
        \item In the first round, $P_1$ plays \emph{wait}. 
        This forces the game to transition to $S_a$ or $S_b$ based on the opponent's action, and then immediately return to $S$.
        \item In the second round, 
        the team players observe the history.
        If the previous state was $S_a$, 
        both $P_1$ and $P_2$ play action $a$, otherwise action $b$. 
    \end{itemize}
        Since team players always choose the same action in the second round,
        the probability of falling into the trap due to choosing different actions is $0$.
        We now calculate the probability of reaching $\Terminals$. 
        Since $\opp$ is memoryless, she plays action $a$ with the same probability
        $p_{\opp,a}$ independently in both the first and second round.
        
        The probability of reaching $\Terminals$ is the sum of the
        probabilities of two disjoint events:
        The history was $S_a$ \emph{and} the opponent plays $a$ again.
        Or the history was $S_b$ \emph{and} the opponent plays $b$ again.
        $$\prob(\Diamond \Terminals) = p_{\opp,a} \cdot p_{\opp,a} + (1-p_{\opp,a}) \cdot (1-p_{\opp,a}) = p_{\opp,a}^2 + (1-p_{\opp,a})^2$$
        To find the guaranteed value, we minimise this function with respect to $p_{\opp,a}$. The function $f(x) = x^2 + (1-x)^2$ achieves its global minimum at $x = 1/2$.$$\min_{p_{\opp,a}} \left( p_{\opp,a}^2 + (1-p_{\opp,a})^2 \right) = \left(\frac{1}{2}\right)^2 + \left(\frac{1}{2}\right)^2 = \frac{1}{4} + \frac{1}{4} = \frac{1}{2}$$
        Thus, the team guarantees reaching probability of at least~$1/2$.
        
        Since $1/2 > 1/4$, there exists a collective strategy with memory that strictly outperforms the best possible memoryless strategy. Which completes the proof.
        
\qed\end{proof}

\subsection{Proof of \Cref{thm:thmETRComplete}} \label{app:ETRComplete}

\thmETRComplete*
We first define an ETR sentence $\Psi$, then prove the correctness. This automatically shows the membership of the problem in $\exists\Rb$.

\subsubsection{Definition of $\Psi$}
Here, we define the formula $\Psi$.

\textbf{Variables.}
We introduce the following existentially quantified variables:
\begin{itemize}
    \item \textbf{Strategy variables:} For each team player $p_i \in \Players$, state $s \in \Vertices$, and action $a \in \Av_i(s)$, let $x_{s,i,a} \in [0,1]$ denote the probability of playing action $a$ by player $p_i$ on state $s$.
    \item \textbf{Valuation variables:} For each state $s \in \Vertices$, let $v_s \in [0,1]$ denote the value of the state $s$ (Which shows the expected payoff of the game, starting from state $s$).
    \item \textbf{Discount variable:} Let $\lambda \in (0,1)$ represent the discount factor.
\end{itemize}

\textbf{The formula $\Psi$.}
The formula is a conjunction of the following constraints:

\begin{enumerate}
    \item \textbf{Strategy constraints:}
    The variables $x$ must define a valid probability distribution on every state:
    \begin{equation}
        \Phi_{\mathsf{strategy}} := \bigwedge_{s \in S} \bigwedge_{i \in \Pi} \sum_{a \in \Av_i(s)} x_{s,i,a} = 1 \land \bigwedge_{a} x_{s,i,a} \ge 0
    \end{equation}

    \item \textbf{Discounted valuation constraints:}
    The values $v(s)$ must effectively be a sub-solution to the $\lambda$-discounted Bellman equations.
    \begin{itemize}
        \item \textit{Target states:} If $s \in T$, the value is fixed to 1:
        \begin{equation}
            \Phi_{\mathsf{target}} := \bigwedge_{s \in \Terminals} 
            v(s) = 1
        \end{equation}
        \item \textit{Non-target states:} If $s \notin T$, the value $v(s)$ must be supportable against \textit{any} action $b$ chosen by the opponent:
        \begin{equation}
            \Phi_{\mathsf{val}} := \bigwedge_{s \in S \setminus T} \bigwedge_{b \in \Av_{\opp}(s)} v(s) \le \lambda \cdot \sum_{\baction \in \Av_{\Pi}(s)} \left( \prod_{i \in \Pi} x_{s,i,a_i} \right) \cdot \sum_{s' \in S} \delta(s, \baction, b)(s') \cdot v(s')
        \end{equation}
    \end{itemize}

    \item \textbf{Threshold constraint:}
    We require the value at the initial state $\initial$ to strictly exceed $t$ and $\lambda$ to be in $(0,\ 1)$:
    \begin{equation}
        \Phi_{\mathsf{goal}} := (0 < \lambda < 1) \land (v(\initial) > t)
    \end{equation}
\end{enumerate}

The complete formula is defined as:
\[ \Psi := \exists \bar{x}, \bar{v}, \lambda : \Phi_{\mathsf{strategy}} \land \Phi_{\mathsf{target}} \land \Phi_{\mathsf{val}} \land \Phi_{\mathsf{goal}} \]

\subsubsection{Correctness of formula $\Psi$}

    The formula $\Psi$ involves polynomials of degree at most 
    $|\Pi|+1$ and the number of variables is polynomial in the size of the game $\mathcal{G}$.
    We establish the equivalence stated in the theorem via the following two lemmas. 
    Thus, the problem would be in the complexity class $\etr$.

\begin{lemma}{\label{lemm:ETR-Soundness}}
    (Soundness) If $\Psi$ is satisfiable,
    then there exists a collective strategy $\bstrat_\Players$ that guarantees
    reaching $\Terminals$ with a probability strictly greater than 
    $\threshold$.
\end{lemma}

\begin{proof}
    Let $(\bar{x}, \bar{v}, \lambda)$ be a satisfying assignment for $\Psi$.
    First, the assignment to variables $\bar{x}$
    defines a valid memoryless collective strategy $\bstrat_\Players$ for the team.
    The constraints in $\Psi$ check that for every state, 
    the probabilities assigned to actions are non-negative and sum to 1.
    
    The variable $v_s$ represents the value of the 
    $\lambda$-discounted game starting at state $s$. 
    Since the assignment satisfies the constraint 
    $v_{\initial} > \threshold$, 
    the expected discounted reward from $\initial$ is greater than $\threshold$.

    We now argue that the collective strategy for the team $\bstrat_\Players$ guarantees
    reaching $\Terminals$ with a probability strictly greater than 
    $\threshold$ in the original game. Consider the payoff for any single play of the game:

    \begin{itemize}
        \item If the game reaches $\Terminals$ after $k$ steps, 
        the reward in the $\lambda$-discounted game is $\lambda^k$. 
        The reward in the original game is $1$.
        \item If the game never reaches $\Terminals$, the reward is 0 in both games.
    \end{itemize}

    In all cases, the payoff in the original game is greater than or equal to the payoff in the $\lambda$-discounted game. Thus, the expected value (probability of winning) in the reachability game is at least as large as the expected reward in the discounted game.

    Since the discounted value guaranteed by the collective strategy $\bstrat_\Players$
    is greater than $\threshold$, $\bstrat_\Players$ guarantees
    reaching $\Terminals$ with a probability strictly greater than 
    $\threshold$.   
\qed\end{proof}

\begin{lemma}{\label{lemm:ETR-Completeness}}
    (Completeness). If there exists a collective strategy guaranteeing a reachability probability greater than $t$, then $\Psi$ is satisfiable.
\end{lemma}

\begin{proof}
    Assume the team can guarantee a reachability probability 
    $V^* > \threshold$. 
    By Theorem \ref{thm:memmless-optimal},
    we know there exists a memoryless collective strategy $\bstrat_\Players^*$ that achieves value $V' > V^*-\epsilon > \threshold$
    for every $\epsilon$ with $V^*-t>\epsilon > 0$. 
    Once this strategy is fixed, 
    the game can be seen as a Markov Decision Process (\emph{MDP}) 
    where the opponent minimises against the fixed team memoryless strategy.
    We briefly recall here that an MDP with reachability objective is the same as a concurrent game where the size of the team players is $0$. 
    
    It is a standard result for finite MDPs that the optimal
    reachability probability is the limit of the $\lambda$-discounted values as $\lambda \to 1^-$ \cite{B62}.
    In other words,
    the reward of reaching $\Terminals$ after $k$ steps
    is $\lambda^k$.
    As $\lambda$ goes to $1^-$, the payoff will reach 1.
    It was shown in \cite{B62}, 
    that (this intuition is actually correct) as $\lambda \to 1^-$, the expected reward is equal to the actual reachability probability.
    
    Let $V^\lambda_{\sigma^*}$ be the value of the game under collective memoryless strategy $\bstrat_\Players^*$ with discount factor $\lambda$. We have $V^*=\lim_{\lambda \to 1^-} V^\lambda_{\sigma^*} \ge V'$.
    
    Since $V'$ is strictly greater than $\threshold$,
    we have $V^*>\threshold$. 
    By the definition of the limit, 
    there must exist a discount factor $\lambda_0 \in (0,1)$ 
    sufficiently close to $1$ such that the discounted value 
    $V^{\lambda_0}_{\sigma^*}$ is also strictly greater than $t$. 
    The tuple $(\sigma^*, V^{\lambda_0}_{\sigma^*}, \lambda_0)$ 
    satisfies all constraints in $\Psi$: $\bstrat_\Players^*$ fits $\Phi_{\mathsf{strategy}}$, $V^{\lambda_0}_{\sigma^*}$ is the unique solution to the contraction mapping for discount $\lambda_0$ fitting $\Phi_{\mathsf{val}}$, and the value exceeds $t$ fitting $\Phi_{\mathsf{goal}}$. Thus, the formula $\Psi$ is satisfiable.
\qed\end{proof}

\subsection{Proof of \Cref{thm:NPhardnessThreshold}} \label{app:NPhardnessThreshold}

\thmNPhardnessThreshold*

We reduce the \textbf{$k$-clique problem} to the threshold problem.
In the $k$-clique problem, we are given a graph $G$ and an integer $k$. We have to determine if there is a set of vertices of size $k$ where every pair is connected by an edge.

We construct a game where the team players can win with probability greater than $\frac{3}{4}$ if and only if such a clique exists.

\subsubsection{Game definition from the clique problem}

Given the graph $G=(V,E)$ and the integer $k$, we construct a game $\mathcal{G}_{G,k}$ as follows.

\paragraph{Players.}
There are two team players, $P_1$ and $P_2$, and one opponent ($\opp$).

\paragraph{States.}
The game has three states: the starting state $s$, the target state $\top$, and the trap state $\bot$.

\paragraph{Actions.}
\begin{itemize}
    \item For the players $P_1$ and $P_2$ the set of available actions at the start state is a pair $(i, v)$, where $i$ is an index (an integer between $1$ and $k$) and $v \in V$ is a vertex:
    $$
        A_{P_1} = A_{P_2} = 
        \{(i, v)\ |\ i \in \{1, 2, ..., k\}, v\in V\}
    $$
    \item $\opp$ chooses a pair of indices $(i^*, j^*)$ between $1$ and $k$.
    $$
        A_\opp = \{(i^*, j^*)\ |\ i^*,j^* \in \{1, 2, ..., k\}\}
    $$
\end{itemize}

\paragraph{Transitions.}
The game starts at state $s$. The transitions are defined as follows:

\begin{enumerate}
    \item \textbf{The Index check (loop):} 
    First, we check if the team players choose the same indices as the opponent.
    Formally, if $P_1$ does not choose index $i^*$ or $P_2$ does not choose index $j^*$, the game stays at state $s$.
    
    \item \textbf{The validation (win/lose):} 
    If the indices match ($P_1$ chooses $i^*$ and $P_2$ chooses $j^*$), we check the vertices they chose ($u$ and $v$).
    \begin{itemize}
        \item \textbf{Consistency:} If the opponent chooses the same indices ($i^* = j^*$), the team passes if $u=v$. 
        If $u \neq v$, the game moves to $\bot$.
        \item \textbf{Edge:} If the opponent chooses different indices ($i^* \neq j^*$), the game moves to $\top$ if there is an edge between them ($(u, v) \in E$). 
        If there is no edge, the game moves to $\bot$.
    \end{itemize}
    
    If the team passes the check, the game moves to $\top$.
\end{enumerate}

Once the game reaches $\top$ or $\bot$, it stays there forever. The team wins if the game reaches $\top$.

We define the threshold $t = \frac{3}{4}$. We show that the team can win with probability greater than $t$ if and only if the graph has a clique of size $k$.

\begin{lemma}
    If the graph $G$ contains a clique of size $k$, the team has a memoryless strategy to reach $\top$ with probability $1$.
\end{lemma}

\begin{proof}
    Assume that the graph $G$ has a clique of size $k$. Let $\{c_1, c_2, \dots, c_k\}$ be the vertices of this clique.

    We define a memoryless collective strategy $\bstrat = (\sigma_{P_1}, \sigma_{P_2})$ for the team players in state $s$ as follows:

    \begin{itemize}
        \item Both players $P_1$ and $P_2$ choose an index $z$ from $1$ to $k$ uniformly at random.
        \item If a player chooses index $z$, they play the action $(z, c_z)$.
    \end{itemize}

    Consider that at any step, some action $(i^*, j^*)$ is chosen by the opponent. 
    In each step, the probability that $P_1$ plays $(i^*, c_{i^*})$ is $1/k$, and the probability that $P_2$ plays $(j^*, c_{j^*})$ is $1/k$. 
    Since the players randomise independently, the probability that they simultaneously match the opponent's indices is $1/k^2$. 
    This probability is strictly positive, so the game will eventually leave the loop.

    When the indices match, we check the transition:
    \begin{itemize}
        \item If $i^* = j^*$, both players play the vertex $c_{i^*}$. Thus, $u = v$, and the game moves to $\top$.
        \item If $i^* \neq j^*$, $P_1$ plays $c_{i^*}$ and $P_2$ plays $c_{j^*}$. Since $c_{i^*}$ and $c_{j^*}$ are distinct vertices in a clique, the edge $(c_{i^*}, c_{j^*})$ exists. Thus, the game moves to $\top$.
    \end{itemize}

    In both cases, the game moves to $\top$. Since the probability of moving to $\bot$ is 0, this strategy guarantees reaching $\top$ with probability 1.
\qed\end{proof}

\begin{lemma}
    If the team has a strategy to reach $\top$ with probability greater than $3/4$, then the graph $G$ contains a clique of size $k$.
\end{lemma}

\begin{proof}
    We assume that the team players have a memoryless collective strategy $\bstrat = (\sigma_{P_1}, \sigma_{P_2})$ that guarantees reaching $\top$ with probability strictly greater than $3/4$. Note that by \cref{thm:memmless-optimal} such a strategy exists. We show that this implies the existence of a clique of size $k$.

    We define $\beta_{P_i}(j)$ as the probability that player $P_i$ chooses the index $j$:
    \[
        \beta_{P_i}(j) = \sum_{v \in V} \sigma_{P_i}(s)(j, v)
    \]

    For the ease of notation, we define $\beta_{P_i}(j, v)$ as the probability of player $P_i$ choosing the pair $(j, v)$. Note that $\beta_{P_i}(j, v) = \sigma_{P_i}(s)(j, v)$.

    First, we argue that for all players $P_i$ and indices $j \in \{1, \dots, k\}$, we must have $\beta_{P_i}(j) > 0$.
    If there exists some $P_i$ and $j$ such that $\beta_{P_i}(j) = 0$, the opponent can choose the index pair $(j, j)$. 
    In this case, the team never matches the opponent's chosen indices. Consequently, the game remains in the starting state $s$ forever. 
    Since the target $\top$ is never reached, the winning probability would be 0, which contradicts the assumption that the winning probability is greater than $3/4$. Thus, $\beta_{P_i}(j) > 0$ for all $i, j$.

    Next, we analyse the conditional probabilities. We define the conditional probability of playing vertex $v$ given index $j$ by player $P_i$ as:
    \[
       \prob(v | j, P_i) = \frac{\beta_{P_i}(j, v)}{\beta_{P_i}(j)}
    \]

    We claim that for every player $P_i$ and index $j$, there must exist a \emph{dominant} vertex $v$ such that $\prob(v | j, P_i) > 1/2$.

    Suppose for the sake of contradiction, for some player $P_1$ and index $j$, no such vertex exists:
    \[
        \max_{v \in V} \frac{\beta_{P_1}(j, v)}{\beta_{P_1}(j)} \leq \frac{1}{2}
    \]

    In this case, suppose the opponent chooses the pair $(j, j)$. The game proceeds to the validation phase only if both $P_1$ and $P_2$ choose action $j$. The probability of passing the consistency check (where $P_1$ and $P_2$ must play the same vertex $v$) is:
    \[
        \prob(\mathsf{Pass} \mid \mathsf{Match}) = \sum_{v \in V} \prob(v | j, P_1) \cdot \prob(v | j, P_2)
    \]

    Since we assumed $\prob(v | j, P_1) \le 1/2$ for every vertex $v$, we can bound this sum:$$    \sum_{v \in V} \prob(v | j, P_1) \cdot \prob(v | j, P_2) \leq \sum_{v \in V} \frac{1}{2} \cdot \prob(v | j, P_2) = \frac{1}{2} \sum_{v \in V} \prob(v | j, P_2)$$

    Since the sum of marginal probabilities for $P_2$ is 1, the probability of passing the consistency check, and reaching $\top$ is at most $1/2$.

    This means that whenever the loop ends, the total probability of winning the game is at most $1/2$. This contradicts our assumption that the team wins with probability greater than $3/4$.

    Therefore, for every player $P_i$ and every index $j$, there must be a dominant vertex $v$ where $\prob(v | j, P_i) > 1/2$. 
    Since these dominant vertices exist, we define a specific sequence of vertices for each player.
    \begin{itemize}
        \item For $P_1$, let $(v_1, \dots, v_k)$ be the sequence where $v_j$ is the dominant vertex for index $j$.
        \item For $P_2$, let $(u_1, \dots, u_k)$ be the sequence where $u_j$ is the dominant vertex for index $j$.
    \end{itemize}

    \textbf{Step 1: The players agree on vertices}.
    We argue that $v_j = u_j$ for all $j$. Suppose they are different ($v_j \neq u_j$). 
    Consider the deterministic strategy by the opponent that chooses $(j, j)$. Then, the probability that the team players play the specific pair $(v_j, u_j)$ is:
    $$    \prob(P_1=v_j, P_2=u_j \mid \text{Match}) > \frac{1}{2} \cdot \frac{1}{2} = \frac{1}{4}$$
    Since $v_j \neq u_j$, the game moves to $\bot$. 
    Thus, the probability of going to $\top$ would be less than $3/4$, which is a contradiction. Thus, $v_j$ must equal $u_j$. Let us call this vertex $c_j$.
    
    \textbf{Step 2: The vertices $(c_1, \ldots, c_k)$ form a clique}.
    We argue that for any two indices $i \neq j$, there is an edge between $c_i$ and $c_j$. Suppose there is no edge between them. 
    Consider the deterministic strategy of the opponent that always chooses $(i, j)$. The team plays the pair $(c_i, c_j)$ with probability greater than $1/4$ (since both are dominant vertices). Because the edge is missing (If $c_i=c_j$, the edge is still missing because the graph does not contain self-loop), the game moves to $\bot$. Again, this implies the winning probability is less than $3/4$, which is a contradiction. 
    Thus, vertices $(c_1, \ldots, c_k)$ form a clique.
\qed\end{proof}

Combining these two lemmas, we conclude the main theorem that the threshold problem is $\NP$-hard.

\section{Appendix for~\cref{sec:almostSure}}

\almostsurememless*
\subsection{Proof of \Cref{thm:almostsurememless}} \label{app:almostsurememless}

We start by defining $W$ as the set of states from which the team players can force the game with an arbitrary strategy to reach $\Terminals$ with probability 1. 
Our goal is to define a memoryless strategy on $W$ that wins almost-surely.

We structure the set $W$ into the sequence $(W_0, W_1, \ldots)$ based on the minimum number of steps that the team can guarantee to reach $\Terminals$ with a positive probability without leaving $W$ (later seen as ranks).
In words, $W_{k+1}$ contains the states that are already in $W_k$, in addition to the states in $W$, from which the team can choose a memoryless strategy that transitions to $W_k$ with positive probability, and remains in $W$ with probability 1, regardless of the opponent's action.
More formally, we define a sequence of sets $W_0, W_1, \dots$ as follows:

\begin{itemize}
    \item $W_0 = \Terminals$.
    \item for $k \geq 0:$ \begin{align*}
    W_{k+1} &= W_k \cup \Bigg\{ s \in W \;\Bigg|\; \exists \bstrat_{\Players} \in \Lambda_{\Players}, \forall b \in \Actions_{\opp} : \\
    &\qquad \left( \sum_{t \in W_k} \sum_{\baction_\Players \in \Av_\Players(s)} \left( \left( \prod_{i \in \Players} \bstrat_i(s)(a_i) \right) \cdot \delta(s, \baction, b)(t) \right) > 0 \right) \\
    &\qquad \land \left( \sum_{t \in W} \sum_{\baction_\Players \in \Av_\Players(s)} \left( \left( \prod_{i \in \Players} \bstrat_i(s)(a_i) \right) \cdot \delta(s, \baction, b)(t) \right) = 1 \right) \Bigg\}
\end{align*}
\end{itemize}

We now prove that this construction eventually includes the set $W$, defined as the set of states from which the team has a collective strategy to reach $T$ almost-surely.

\begin{restatable}{lemma}{almostsurewconvergence} \label{lemm: almost-sure-w-convergence}
    Let $W$ be the set of all almost-sure winning states. The sequence $W_k$ converges to $W$. That is, there exists $m$ such that $W_m = W$.
\end{restatable}

\begin{proof}
    Since the set of states $\Vertices$ is finite, 
    the sequence $W_0 \subseteq W_1 \subseteq \dots$ 
    must eventually stop growing. 
    Let $W^*$ be the limit of this sequence. 
    By definition, $W^* \subseteq W$ because only states from $W$ are added to set $W_{k+1}$ at each step $k$.
    
    Suppose for the sake of contradiction, 
    the set $W^*$ is not equal to $W$. 
    This means there is a non-empty set of states 
    $R = W \setminus W^*$. 
    From any state in $R$, the team can win with probability 1. 
    However, to reach $\Terminals$, 
    the game must eventually leave $R$ and enter $W^*$ 
    (since $W^*$ contains $\Terminals$). 

    Since $R \subseteq W$, 
    there exists a collective strategy $\bstrat_{\Pi}$ 
    for the team that guarantees almost-sure reachability from any state $s \in R$.
    However, by the definition of the limit set $W^*$, 
    for every state $s \in R$, 
    the team cannot force a transition to $W^*$ with positive probability without leaving the set $W$, because if there existed a memoryless strategy in $s$ that forces entering $W^*$ with positive probability without leaving $W$, then $s$ would be added to $W_{k+1}$ once $W_k = W^*$, hence $s \in W^*$, which contradicts the fact that $s \in R$.
    Therefore, even if the team plays $\bstrat_{\Pi}$, 
    the opponent can force the game to remain outside $W^*$ at every step.
    Since $\Terminals \subseteq W^*$, this opponent strategy guarantees the game never reaches $\Terminals$, contradicting the assumption that $\bstrat_{\Pi}$ is an almost-sure winning strategy.
    
    Therefore, the set $R$ must be empty, 
    which proves that $W^* = W$.
\qed\end{proof}

Now, we construct the strategy $\eta$ based on the sequence $W_i$, which guarantees almost-sure reachability. 
For every state $s \in W \setminus \Terminals$,
we define $r_s$ the smallest index
such that $s \in W_{r_s}$. 
We call $r_s$ the rank of $s$.
Let $m$ be the maximum rank (i.e., $W_m = W$).

By the definition of the sequence $W_i$,
there exists a memoryless strategy 
for the team at state $s$ such that for any opponent action:

\begin{enumerate}
    \item The game moves to $W_{k-1}$ with a positive probability on the next turn.
    \item The game remains in $W$ with probability 1 on the next turn.
\end{enumerate}

Let $\eta$ be the strategy that plays this memoryless strategy at each state $s$. We show in the following lemma that for each state $s \in W$, strategy $\eta$ guarantees reaching $\Terminals$ with probability 1.

\begin{restatable}{lemma}{almostsurememlesss} \label{lemm:almost-sure-strat}
    If the team players play the memoryless collective strategy $\eta$, for every state $s \in W$, the probability of reaching $\Terminals$ is 1.
\end{restatable}

\begin{proof}
    Since $W$ is a finite set and the probability of reducing the rank is strictly positive at each step, there exists a value $\gamma > 0$ such that for any state $s \in W$, the probability of reaching $\Terminals$ (rank 0) within at most $m$ steps is at least $\gamma$ by using the strategy $\eta$.

    Now, let $p_{\neg \text{reach}}(i)$ 
    denote the probability that the strategy $\eta$ 
    fails to reach $\Terminals$ within the first $i$ steps.
    In the first $m$ steps, 
    the probability of reaching $\Terminals$ is at least $\gamma$. 
    If the target is not reached, 
    the game remains in $W$, 
    and the same argument applies to the next $m$ steps.
    Therefore, for $i = k \cdot m$, this probability is bounded by:
    \[
        p_{\neg \text{reach}}(k \cdot m) \leq (1 - \gamma)^k
    \]
    
    As $k \to \infty$, the probability
    $p_{\neg \text{reach}}(k \cdot m)$
    converges to 0. Consequently, the probability of eventually reaching $\Terminals$ is:
    \[
        1 - \lim_{k \to \infty} p_{\neg \text{reach}}(k \cdot m)  = 
        1 - \lim_{k \to \infty} (1 - \gamma)^k = 1.
    \]
    Thus, $\eta$ is an optimal memoryless almost-sure winning strategy.
\qed\end{proof}

Combining \cref{lemm: almost-sure-w-convergence} and \cref{lemm:almost-sure-strat}, we establish the main theorem for this section that shows memoryless strategies are optimal in the almost-sure problem.

\subsection{Proof of \Cref{thm:almostsureeasiness}} \label{app:almostsuresat}

\almostsureeasiness*

We construct the SAT formula $\Phi$ that is satisfiable if and only if there exists a winning set $W$, a ranking $r_s$, and a strategy support that guarantee almost-sure winning as follows.

\subsubsection{Variables}

Let $N = |\Vertices|$ be the number of states. We define the following Boolean variables:
\begin{itemize}
    \item \textbf{Winning set}: For each state $s \in \Vertices$, a variable $w_s$ is defined. This variable will be true if $s \in W$.
    \item \textbf{Ranking}: For each state $s \in \Vertices$ and rank $k \in \{0, \dots, N\}$, a variable $r_{s,k}$ is defined. 
    This variable will be true if $r_s$ is equal to $k$ 
    (i.e., $r_s$ is the smallest index which $s \in w_{r_s}$).
    \item \textbf{Strategy support}: For each state $s$, player $P_i \in \Players$, and action $a \in \Actions_{P_i}$, a variable $u_{s,i,a}$ is defined. This variable will be true if the action $a$ is in the support of the player $i$ at the state $s$ (i.e. $a$ is played with positive probability).
\end{itemize}

\subsubsection{Constraints.} We define the following constraints:

\textbf{1. Consistency of winning set and ranks.}
A state $s$ is in the winning set $W$ if and only if it has a rank.
Furthermore, we force that each state has \emph{at most} one rank. 
These two ensure that every state in $W$ has exactly one rank. 
Target states must have rank 0. 
We also force $\initial$ to be in $W$.
\begin{align*}
    \Phi_{\text{rank}} :=
    \bigwedge_{s \in \Vertices} \left( w_s \leftrightarrow \bigvee_{k=0}^N r_{s,k} \right) 
    \bigwedge_{s \in \Vertices} \bigwedge_{0 \le j < k \le N} \neg (r_{s,j} \land r_{s,k}) 
    \bigwedge_{s \in \Terminals} r_{s,0}
    \ \ \bigwedge w_{\initial}
\end{align*}

\textbf{2. Non-Empty strategy support.}
If a state is in $W$, every team player must have at least one valid action in their support.
\[
    \Phi_{\text{support}} := \bigwedge_{s \in \Vertices} \bigwedge_{P_i \in \Players} \left( w_s \to  \bigvee_{a \in \Actions_i} u_{s,i,a} \right)
\]

\textbf{3. Safety in $W$.}
The strategy must guarantee that the game never leaves $W$. 
If the team plays actions from their supports, 
then for any opponent action $b$, 
all possible next states must be in $W$.
\[
    \Phi_{\text{safe}} := \bigwedge_{s \in \Vertices} \bigwedge_{\baction \in \prod\limits_{i} Av_{i}(s)} \bigwedge_{b \in \Actions_{\opp}}
    \left( \left(\bigwedge_{P_i \in \Players} u_{s,i,a_i} \right) \to \bigwedge_{t \in \text{Supp}(\delta(s, \baction, b))} w_t \right)
\]

\textbf{4. Progress (rank reduction).}
From $s \in W$ with rank $k > 0$, the strategy must have a positive probability to move to a state with a lower rank. 
Formally, for every opponent action $b$, there must be a joint action in the team's support that leads to a state $t$ with rank strictly less than $k$. $\Phi_{\text{prog}}=$
\[
      \bigwedge_{s \in \Vertices} \bigwedge_{k=1}^N \bigwedge_{b \in \Actions_{\opp}} \left( r_{s,k} \to \bigvee_{\baction \in \prod\limits_{i} Av_{i}(s)}
    \left( 
       \left( \bigwedge_{P_i \in \Players} u_{s,i,a_i} \right) \land \bigvee_{t \in \text{Supp}(\delta(s, \baction, b))} \bigvee_{j=0}^{k-1} r_{t,j}
    \right) \right)
\]

\subsubsection{Proof of correctness}

Based on \cref{sec:almost-sure-memoryless}, we know that if an almost-sure winning strategy exists, there exists a memoryless strategy that guarantees almost-sure reachability. 

First, we argue in the following lemma that the almost-sure reachability depends only
on the \emph{support} of the memoryless strategies rather than their precise numeric values.
The support of a strategy on a state is defined as the set of actions that are played with a positive probability.

\begin{restatable}{lemma}{almostsuresupport} \label{lemm:almost-sure-support}
    Let $\sigma$ and $\sigma'$ be two memoryless collective strategies for the team. If $\sigma$ and $\sigma'$ have the same support (i.e., for all $s \in S$ and $P_i \in \Pi$, $\{a \in A_i \mid \sigma_i(s)(a) > 0\} = \{a \in A_i \mid \sigma'_i(s)(a) > 0\}$), then $\sigma$ guarantees almost-sure reachability to $\Terminals$ if and only if $\sigma'$ does.
\end{restatable}

\begin{proof}
    Fix a memoryless collective strategy $\sigma$ for the team players. 
    Since $\sigma$ is fixed, the game reduces to a \emph{Markov Decision Process} (\emph{MDP}) where only the opponent chooses actions.

    In this MDP, the opponent's goal is to minimize the probability of reaching $\Terminals$. 
    The transitions are defined by the opponent's action $b \in \Actions_{\opp}$ and $\sigma$. Specifically, a transition from $s$ to $t$ under action $b$ is possible (has probability $>0$) if and only if the team plays a joint action $\baction$ with a positive probability in $\sigma$ such that $\delta(s, \baction, b)(t) > 0$.

    A standard result in MDPs is that the player who minimises (opponent in our setting) can prevent almost-sure reachability if and only if they can force the game into a set of non-target states called an \emph{end component} \cite{puterman1994markov}. An end component is a set of states where opponent can choose actions to keep the game inside the set forever. 
    More formally, 
    an end component is a set of states $EC$ such that for every state $s \in EC$, 
    there exists an opponent action $b$ where all possible successor states belong to $EC$.

    The existence of an end component depends only on the graph structure 
    (the set of edges with non-zero probability), 
    not on the exact transition probabilities \cite{puterman1994markov}. 
    Since $\sigma$ and $\sigma'$ have the same support, 
    they make the same transitions available for the opponent. 
    Thus, the graph structures of their induced MDPs are the same.
    Therefore, the opponent can avoid reaching $\Terminals$ with probability 1 under $\sigma$ if and only if they can do so under $\sigma'$.
\qed\end{proof}

By \cref{lemm:almost-sure-support}, the existence of an almost-sure winning strategy depends only on the support of the strategy. 
Therefore, the problem reduces to finding the support of a strategy that guarantees almost-sure reachability.

\begin{restatable}{lemma}{almostsurephi} \label{lemm:almostsurephi}
    The formula $\Phi$ is satisfiable if and only if
    there exists a memoryless strategy $\eta$ that guarantees
    reaching the target $\Terminals$  from state $\initial$ with probability 1.
\end{restatable}

\begin{proof}
    If $\Phi$ is satisfiable, 
    the variables $u_{s,i,a}$ 
    define the support of a memoryless collective strategy $\eta$. 
    The Safety constraint ensures the team stays in $W$, 
    and the Progress constraint ensures that from any state in $W$, the game moves to a lower rank with a positive probability. 
    As stated in \cref{sec:almost-sure-memoryless}, 
    this guarantees reaching rank 0 ($\Terminals$) with probability 1.
    
    On the other hand, suppose there exists a memoryless collective strategy $\eta$ that wins almost-surely. We construct a satisfying assignment for $\Phi$ as follows:
    \begin{enumerate}
        \item Let $W$ be the set of almost-sure winning states under $\eta$. Set $w_s = \text{true}$ for all $s \in W$.
        \item Define the ranks based on the layers $W_k$ constructed in \cref{sec:almost-sure-memoryless}. For each $s \in W$, let $k$ be the smallest index such that $s \in W_k$. Set $r_{s,k} = \text{true}$ and all other rank variables for $s$ to false.
        \item Set $u_{s,i,a} = \text{true}$ if and only if $\eta_i(s)(a) > 0$.
    \end{enumerate}
    This assignment satisfies all constraints by construction. 
    The Consistency constraints hold because every state in $W$ belongs to exactly one layer $W_k$.
    The Non-Empty Support holds because $\eta$ is a valid strategy.
    The Safety constraint is satisfied because an almost-sure winning strategy cannot leave the winning set $W$ (otherwise there would be a non-zero probability of never reaching $\Terminals$).
    The Progress constraint is satisfied by the definition of the sequence $(W_0, W_1, \ldots)$: for any $s \in W_k$ ($k>0$), the strategy $\eta$ guarantees a transition to $W_{k-1}$ with positive probability against any opponent action.
\qed\end{proof}

From \cref{lemm:almostsurephi} we conclude with the main theorem.

\subsubsection{Binary rank encoding}

The encoding presented above uses $\Ocomplexity(N)$ unary variables to represent the rank of each state, leading to a total of $\Ocomplexity(N^2)$ variables. To improve scalability in the experiments, we replace this with a binary rank encoding. We use $L = \Ocomplexity(\log{N})$ Boolean variables per state to represent its rank $r_s$ in binary form. This reduces the total number of rank variables to $\Ocomplexity(N \log N)$.

We adapt the progress constraint $r_t < r_s$ using a bitwise comparator circuit. Let $\bar{b}_t$ and $\bar{b}_s$ be the binary vectors for the ranks of states $t$ and $s$. We encode the strictly-less-than relation as:
\[
    \text{BSlt}(\bar{b}_t, \bar{b}_s) := \bigvee_{i=0}^{L-1} \left( \neg b_{t,i} \land b_{s,i} \land \bigwedge_{j=i+1}^{L-1} (b_{t,j} \leftrightarrow b_{s,j}) \right)
\]
This predicate represents the condition that at the most significant bit position $i$ where the vectors differ, the bit in $\bar{b}_t$ is 0 and the bit in $\bar{b}_s$ is 1. This new constraint uses $\Ocomplexity(\log N)$ literals, providing an $\tilde{\Ocomplexity}(N)$ improvement in the number of variables.

This formulation simplifies the progress constraint significantly. In the unary encoding, $\Phi_{\text{prog}}$ required an explicit disjunction over all possible ranks $k \in \{1, \dots, N\}$. With the binary comparator, this disjunction is replaced by a single symbolic evaluation.

We rewrite the progress constraint as follows:

\begin{align*}
\Phi_{\text{prog}} = \bigwedge_{s \in \Vertices \setminus \Terminals} \bigwedge_{b \in \Actions_{\opp}} \bigg( w_s \to & \bigvee_{\baction \in \prod_{i} Av_{i}(s)} \Big( \big( \bigwedge_{P_i \in \Players} u_{s,i,a_i} \big) \\
& \land \bigvee_{t \in \text{Supp}(\delta(s, \baction, b))} \text{BSlt}(\bar{b}_t, \bar{b}_s) \Big) \bigg)
\end{align*}

By eliminating the iteration over $k$, the size of the formula for each transition is reduced by a factor $\tilde{\Ocomplexity}(N)$. We use this optimisation in the experiments.

\section{Appendix for~\cref{sec:logic}} \label{app:logic}

\begin{theorem}\label{thm:safetyAS}
    Given a game $\Game$, if there is a collective strategy $\bstrat_\Pi$ that ensures that the target set $T$ is avoided with probability $1$, then there is a deterministic one. 
\end{theorem}
\begin{proof}
    Let $S_{\text{safe}}$ be the set of states in the game from which there is a collective strategy that ensures that the target $T$ is avoided with probability $1$. Clearly, observe that $S_{\text{safe}}\cap T = \emptyset$.

    We construct a deterministic memoryless collective strategy $\bstrat^\text{det}_\Players$ for the team on the set $S_{\text{safe}}$ as follows.
Let $\bstrat_\Players$ be the collective strategy witnessing that a state $s \in S_{\text{safe}}$ is safe. 
At any state $s \in S_{\text{safe}}$, the collective strategy defines a joint distribution over actions $\mathcal{D}_s = \prod_{p \in \Players} \strat_p(s)$. Informally, this is the distribution proposed by the strategy assuming the play starts at the state $s$.   
Since this collective strategy ensures safety with probability~$1$, the support of this distribution must contain actions that also lead to the set~$S_{\text{safe}}$.

We choose a joint action profile $\baction^* = (a^*_p)_{p \in \Players} \in \prod_{p \in \Players} \Av_p(s)$ such that:
\[
    \mathcal{D}_s(\baction^*) > 0.
\]
We then define the deterministic memoryless strategy $\bstrat^\text{det}_\Players$ such that for every player $p \in \Players$ and state $s \in S_{\text{safe}}$, $\strat^\text{det}_p(s)$ assigns probability $1$ to the action $a^*_p$.

We now show that this strategy $\bstrat^\text{det}_\Players$ ensures $T$ is avoided with probability~$1$. 
Consider any state $s \in S_{\text{safe}}$ and the chosen joint action $\baction^*$. 
Let $b \in\Av_\opp(s)$ be any action profile for the opponent.
Let $t$ be any state in the support of the transition function $\delta(s, (\baction^*, b))$.

We claim that $t \in S_{\text{safe}}$. 
Suppose for the sake of contradiction that $t \notin S_{\text{safe}}$. By the definition of $S_{\text{safe}}$, this would imply that from $t$, the opponents have a strategy to force the game into $T$ with non-zero probability.
However, since $\baction^*$ is played with positive probability under the original safe strategy $\bstrat_\Players$, and $t$ is reached with positive probability given $\baction^*$ and $\mathbf{b}$, it follows that the original strategy $\bstrat_\Players$ allows the game to reach a state ($t$) from which $T$ is reached with positive probability. This implies that the probability of avoiding $T$ from $s$ under $\bstrat_\Players$ is strictly less than $1$, contradicting the assumption that $s \in S_{\text{safe}}$.

Thus, all successors of $s$ under the strategy $\bstrat^\text{det}_\Players$ (regardless of the opponent's move) remain in $S_{\text{safe}}$. Since the initial state is in $S_{\text{safe}}$ and $S_{\text{safe}} \cap T = \emptyset$, the team avoids $T$ forever with probability $1$. Observe that since the strategy is deterministic, no play of the game on this strategy avoids the target set $T$ and not just plays of measure $1$. 
\qed\end{proof}
\begin{corollary}\label{cor:safetyAS-is-Sure}
     Given a game $\Game$, if there is a collective strategy $\bstrat_\Pi$ that ensures that the set $T$ is avoided with probability $1$, then there is a deterministic strategy that ensures \emph{all} plays compatible with $\bstrat_\Pi$ avoid the set $T$. 
\end{corollary}
\begin{proof}
    This follows from the fact that in one-step, all one-step successors of the state $ S_{\text{safe}}$, chosen by the deterministic strategy constructed above also land in the set $ S_{\text{safe}}$. All plays therefore stay in $ S_{\text{safe}}$. 
\qed\end{proof}

\section{Appendix for~\cref{sec:experiments}}\label{app:experiments}
\subsection{Optimisation strategy for local reachability games}\label{app:slsqp}

This section explains how we solve the one-shot games in the value iteration algorithm heuristically. 
In the value iteration methods,
the main challenge is calculating the $Pre$ operator at each step.
This operator tells us the best probability of winning the team can guarantee in a single step, 
assuming the opponent plays optimally.

Because the team players randomise independently, we cannot simply use linear programming, which is the
typical strategy to solve a zero-sum two-player concurrent game.
Instead, we have to solve a nonlinear optimisation problem, which is explained below. 
We use the Sequential Least Squares Programming (SLSQP) algorithm for this problem. 
Below, we explain the formal logic behind the problem and why SLSQP can work well.

\subsubsection{Optimisation problem}
To find the value $Pre(u_k)(s)$ at a specific step $k$, 
we have to solve the following single maximisation problem.
We want to find the best probability distributions for the team ($\mathbf{p}$) and a value $v$ such that the team wins with at least probability $v$, 
no matter what the opponent does. 
We formulate this as:

$$\begin{aligned}
\text{maximize} \quad & v \\
\text{subject to} \quad & \sum_{a \in A_i} \mathbf{p}_i(a) = 1, \quad \forall i \in \Pi \\
& \mathbf{p}_i(a) \geq 0, \quad \forall i \in \Pi, \forall a \in A_i \\
& \sum_{\vec{a} \in A_\Pi} \left( \prod_{i \in \Pi} \mathbf{p}_i(a_i) \right) \cdot \mathcal{V}_{s,k}(\vec{a}, b) \geq v, \quad \forall b \in A_O
\end{aligned}$$

where $\mathcal{V}_{s,k}(\vec{a}, b) = \sum_{s' \in S} \delta(s, \vec{a}, b)(s') \cdot u_k(s')$ is the expected next-step value given joint team action $\vec{a}$ and opponent action $b$ based on the previous values of the value iteration algorithm.

In this setup, the objective is simple 
(just maximise $v$). 
The constraints on the probabilities are also simple linear equations. 
The only hard part is the last constraint: $\mathcal{V}_{s,k}(\vec{a}, b)$ depends on the product of the team's probabilities, which makes it a polynomial function. Thus, we cannot solve it simply with a Linear Programming (LP) solver. 
Instead, we try to solve the optimisation problem with a heuristic approach instead of using slow SMT solvers (the VI-ETR algorithm). 

\subsubsection{Smoothness and differentiability}

To use fast optimisation algorithms, we need the problem to be smooth. In mathematical terms, this means the functions should be infinitely differentiable ($C^\infty$).

Our problem satisfies this. The constraints are just polynomials (sums and products of probabilities). Polynomials are smooth everywhere. Thus, we can calculate their slopes (gradients) and curvatures (Hessians) exactly. This allows us to use advanced algorithms that rely on this curvature information to find the answer quickly, rather than guessing blindly.

\subsubsection{Algorithm Selection: SLSQP}

We use the SLSQP algorithm to solve this optimisation problem. While there are many non-linear solvers available, we chose SLSQP for the following specific reasons, motivated by our problem structure and results from similar fields:

\paragraph{1. Efficient handling of linear constraints.}
Our basic constraints are that probabilities must sum to 1 and be non-negative. 
These are linear constraints. 
Unlike penalty-based methods that may explore infeasible regions during intermediate steps, 
SLSQP exploits the structure of linear constraints to satisfy them exactly at every iteration~\cite{diehl2009efficient}. 
Since the geometry of our search space is defined largely by the probability simplex, 
this ensures the solver never wastes computational effort evaluating undefined or invalid probability distributions 
(e.g., negative probabilities or not with a sum equal to 1).

\paragraph{2. Efficiency with sparse strategies.}
In concurrent games, the best strategies are often simple: players sometimes rely on just a small number of actions rather than mixing many options. This means the probability for most actions should be exactly zero. SLSQP employs an \textbf{active-set strategy} which is designed to quickly identify and lock onto these zero values. This contrasts with Interior Point methods, which approach zero only gradually and can face numerical precision issues when the optimal solution requires exact zeros~\cite{diehl2009efficient}.

\paragraph{3. Proven performance in dynamic games.}
SQP-based methods are a standard choice for solving games that evolve over time. For instance, Zhu and Borrelli~\cite{Zhu2022ASQ} showed that SQP acts much faster than generic solvers when finding solutions to dynamic games. They found that SQP is particularly good at handling the tight interactions between players. This performance improves further when the solver uses exact gradients instead of estimating them, which is the technique used in SLSQP.

\paragraph{4. Sound under-approximation.}
Since our optimisation problem is non-convex, 
SLSQP behaves as a local optimiser. 
This guarantees that the value $v^*$ 
it returns is a valid \textbf{under-approximation} 
of the true maximum ($v^* \leq Pre(u_k)(s)$). 
In the context of formal verification, this is a ``safe'' error: we might underestimate the team's winning chance, 
but we will never falsely claim a higher probability than is actually achievable. 
Our \texttt{VI-Hybrid} algorithm balances this error by using the SLSQP result as a fast, 
high-quality starting point to help the exact SMT solver find the global optimum.

\subsection{Benchmarks}
\label{app:experiments_benchmarks}

In this section, we provide a formal definition for each of the benchmarks in the paper. 

\subsubsection{Pursuit-Evasion with rendezvous.} We analyze a variant of the pursuit-evasion game \cite{parsons2006pursuit} played on a directed graph $G=(V, E)$. The game contains cooperative team players $\Players = \{P_1, \dots, P_k\}$ against a single opponent $\opp$. A global state is defined by the position vector $s = (l_1, \dots, l_k, l_{\opp})$, and all players have perfect information.

\textbf{Dynamics.} Rounds are concurrent. From a current node $u$, each player independently chooses a next step from the closed neighbourhood $N(u) = \{v \mid (u, v) \in E\} \cup \{u\}$. This allows agents to either move to an adjacent node or wait at their current position. Players can play a mixed strategy over their actions.

\textbf{Objectives.} The team aims to meet at a single node while avoiding the opponent. The conditions are:
\begin{enumerate}
    \item \textbf{Rendezvous (win):} All team members meet at the same node ($l_1 = \dots = l_k$).
    \item \textbf{Capture (loss):} The opponent intercepts any team member ($\exists i \text{ s.t. } l_i = l_{\opp}$).
\end{enumerate}

Capture strictly takes precedence over rendezvous. If the team meets at a node also occupied by the opponent, it counts as a loss. Consequently, the target set $T$ consists only of safe rendezvous states (where $l_1 = \dots = l_k \neq l_{\opp}$). The team maximises the probability of reaching $T$, while the opponent minimises it.

We evaluated six scenarios shown in Figure \ref{fig:gamegraphs}. The results are presented in Table \ref{table:graphgame}.

\begin{figure}
\begin{center}
\begin{tikzpicture}[>=stealth, node distance=2cm, scale=0.60, transform shape]
  \tikzstyle{state}=[circle, draw, minimum size=0.8cm, fill=white]
  \begin{scope}[shift={(0,0)}]
    \node[font=\bfseries] at (0.00, 3.0) {Scenario 1};
    \node[state, fill=green!20, label={[fill=white, fill opacity=0.8, inner sep=1pt]-146:$T_{1}$}] (Scenario1_0) at (-1.50, -1.50) {0};
    \node[state, fill=green!20, label={[fill=white, fill opacity=0.8, inner sep=1pt]-33:$T_{2}$}] (Scenario1_1) at (1.50, -1.50) {1};
    \node[state, fill=red!20, label={[fill=white, fill opacity=0.8, inner sep=1pt]90:$Opp$}] (Scenario1_2) at (0.00, 1.50) {2};
    \draw[->] (Scenario1_0) edge[bend right=15] (Scenario1_2);
    \draw[->] (Scenario1_1) edge[bend right=15] (Scenario1_0);
    \draw[->] (Scenario1_1) edge[bend right=15] (Scenario1_2);
    \draw[->] (Scenario1_2) edge[bend right=15] (Scenario1_0);
    \draw[->] (Scenario1_2) edge[bend right=15] (Scenario1_1);
  \end{scope}
  \begin{scope}[shift={(6,0)}]
    \node[font=\bfseries] at (0.00, 3.0) {Scenario 2};
    \node[state] (Scenario2_0) at (-1.50, 1.50) {0};
    \node[state, fill=green!20, label={[fill=white, fill opacity=0.8, inner sep=1pt]-135:$T_{1}$}] (Scenario2_1) at (-1.50, -1.50) {1};
    \node[state, fill=red!20, label={[fill=white, fill opacity=0.8, inner sep=1pt]-45:$Opp$}] (Scenario2_2) at (1.50, -1.50) {2};
    \node[state, fill=green!20, label={[fill=white, fill opacity=0.8, inner sep=1pt]45:$T_{2}$}] (Scenario2_3) at (1.50, 1.50) {3};
    \draw[->] (Scenario2_0) edge[bend right=15] (Scenario2_1);
    \draw[->] (Scenario2_0) edge[bend right=15] (Scenario2_2);
    \draw[->] (Scenario2_0) edge[bend right=15] (Scenario2_3);
    \draw[->] (Scenario2_1) edge[bend right=15] (Scenario2_0);
    \draw[->] (Scenario2_1) edge[bend right=15] (Scenario2_2);
    \draw[->] (Scenario2_2) edge[bend right=15] (Scenario2_0);
    \draw[->] (Scenario2_2) edge[bend right=15] (Scenario2_1);
    \draw[->] (Scenario2_2) edge[bend right=15] (Scenario2_3);
    \draw[->] (Scenario2_3) edge[bend right=15] (Scenario2_0);
    \draw[->] (Scenario2_3) edge[bend right=15] (Scenario2_2);
  \end{scope}
  \begin{scope}[shift={(12,0)}]
    \node[font=\bfseries] at (0.00, 3.0) {Scenario 3};
    \node[state] (Scenario3_0) at (0.00, 0.00) {0};
    \node[state, fill=green!20, label={[fill=white, fill opacity=0.8, inner sep=1pt]143:$T_{1}$}] (Scenario3_1) at (-1.50, 1.50) {1};
    \node[state, fill=red!20, label={[fill=white, fill opacity=0.8, inner sep=1pt]-90:$Opp$}] (Scenario3_2) at (0.00, -1.50) {2};
    \node[state, fill=green!20, label={[fill=white, fill opacity=0.8, inner sep=1pt]36:$T_{2}$}] (Scenario3_3) at (1.50, 1.50) {3};
    \draw[->] (Scenario3_1) -- (Scenario3_0);
    \draw[->] (Scenario3_1) -- (Scenario3_2);
    \draw[->] (Scenario3_2) -- (Scenario3_0);
    \draw[->] (Scenario3_3) -- (Scenario3_0);
    \draw[->] (Scenario3_3) -- (Scenario3_2);
  \end{scope}
  \begin{scope}[shift={(0,-6)}]
    \node[font=\bfseries] at (0.00, 3.0) {Scenario 4};
    \node[state, fill=green!20, label={[fill=white, fill opacity=0.8, inner sep=1pt]147:$T_{1}$}] (Scenario4_0) at (-1.50, 1.20) {0};
    \node[state, fill=red!20, label={[fill=white, fill opacity=0.8, inner sep=1pt]90:$Opp$}] (Scenario4_1) at (0.00, 1.20) {1};
    \node[state, fill=green!20, label={[fill=white, fill opacity=0.8, inner sep=1pt]32:$T_{2}$}] (Scenario4_2) at (1.50, 1.20) {2};
    \node[state] (Scenario4_3) at (-0.90, -1.20) {3};
    \node[state] (Scenario4_4) at (0.90, -1.20) {4};
    \draw[->] (Scenario4_0) -- (Scenario4_3);
    \draw[->] (Scenario4_0) -- (Scenario4_4);
    \draw[->] (Scenario4_1) -- (Scenario4_3);
    \draw[->] (Scenario4_1) -- (Scenario4_4);
    \draw[->] (Scenario4_2) -- (Scenario4_3);
    \draw[->] (Scenario4_2) -- (Scenario4_4);
  \end{scope}
  \begin{scope}[shift={(6,-6)}]
    \node[font=\bfseries] at (0.00, 3.0) {Scenario 5};
    \node[state, fill=green!20, label={[fill=white, fill opacity=0.8, inner sep=1pt]135:$T_{1}$}] (Scenario5_0) at (-1.50, 1.50) {0};
    \node[state, fill=green!20, label={[fill=white, fill opacity=0.8, inner sep=1pt]45:$T_{2}$}] (Scenario5_1) at (1.50, 1.50) {1};
    \node[state, fill=green!20, label={[fill=white, fill opacity=0.8, inner sep=1pt]-135:$T_{3}$}] (Scenario5_2) at (-1.50, -1.50) {2};
    \node[state, fill=red!20, label={[fill=white, fill opacity=0.8, inner sep=1pt]-45:$Opp$}] (Scenario5_3) at (1.50, -1.50) {3};
    \draw[->] (Scenario5_0) edge[bend right=15] (Scenario5_1);
    \draw[->] (Scenario5_0) edge[bend right=15] (Scenario5_2);
    \draw[->] (Scenario5_0) edge[bend right=15] (Scenario5_3);
    \draw[->] (Scenario5_1) edge[bend right=15] (Scenario5_0);
    \draw[->] (Scenario5_1) edge[bend right=15] (Scenario5_2);
    \draw[->] (Scenario5_1) edge[bend right=15] (Scenario5_3);
    \draw[->] (Scenario5_2) edge[bend right=15] (Scenario5_0);
    \draw[->] (Scenario5_2) edge[bend right=15] (Scenario5_1);
    \draw[->] (Scenario5_2) edge[bend right=15] (Scenario5_3);
    \draw[->] (Scenario5_3) edge[bend right=15] (Scenario5_0);
    \draw[->] (Scenario5_3) edge[bend right=15] (Scenario5_1);
    \draw[->] (Scenario5_3) edge[bend right=15] (Scenario5_2);
  \end{scope}
  \begin{scope}[shift={(12,-6)}]
    \node[font=\bfseries] at (0.00, 3.0) {Scenario 6};
    \node[state] (Scenario6_4) at (0.00, 1.95) {4};
    \node[state, fill=green!20, label={[fill=white, fill opacity=0.8, inner sep=1pt]141:$T_{1}$}] (Scenario6_0) at (-1.50, 1.20) {0};
    \node[state, fill=green!20, label={[fill=white, fill opacity=0.8, inner sep=1pt]38:$T_{2}$}] (Scenario6_1) at (1.50, 1.20) {1};
    \node[state, fill=green!20, label={[fill=white, fill opacity=0.8, inner sep=1pt]-141:$T_{3}$}] (Scenario6_2) at (-1.50, -1.20) {2};
    \node[state, fill=red!20, label={[fill=white, fill opacity=0.8, inner sep=1pt]-38:$Opp$}] (Scenario6_3) at (1.50, -1.20) {3};
    \node[state] (Scenario6_5) at (0.00, -1.95) {5};
    \draw[->] (Scenario6_0) -- (Scenario6_4);
    \draw[->] (Scenario6_0) -- (Scenario6_5);
    \draw[->] (Scenario6_1) -- (Scenario6_4);
    \draw[->] (Scenario6_1) -- (Scenario6_5);
    \draw[->] (Scenario6_2) -- (Scenario6_4);
    \draw[->] (Scenario6_2) -- (Scenario6_5);
    \draw[->] (Scenario6_3) -- (Scenario6_4);
    \draw[->] (Scenario6_3) -- (Scenario6_5);
  \end{scope}
\end{tikzpicture}
\end{center}
\caption{Visualizations of the six scenarios used for evaluation. Scenarios 1--4 involve two team players. Scenarios 5 and 6 involve three team players. Green nodes indicate the team's starting positions, while the red node indicates the opponent's. The results are available in Table \ref{table:graphgame}.}
\label{fig:gamegraphs}
\end{figure}
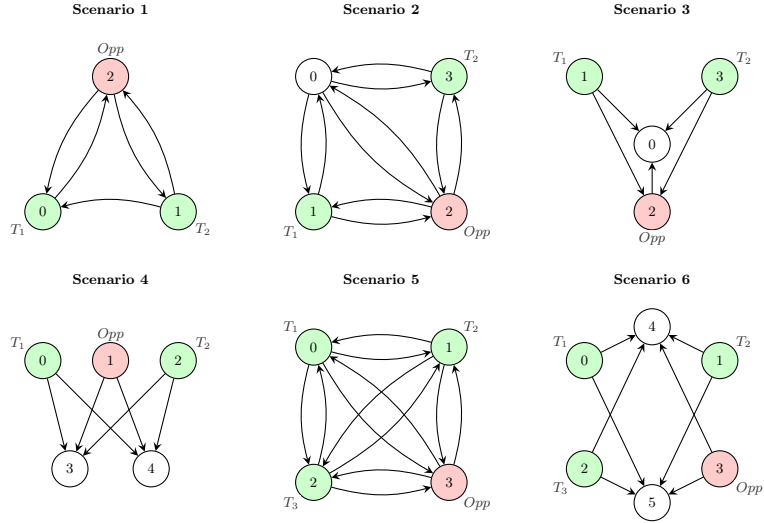

\input{ExpeirmentFigures/almost_sure_graph_scenarios}

\subsubsection{Robot coordination on a grid.}
This benchmark is adapted from the robot coordination case study in~\cite{kwiatkowska2021automatic}. The game occurs on an $H \times W$ grid. The cooperative team consists of $k$ robots $\Players = \{r_1, \dots, r_k\}$, while the opponent $\opp$ controls the environment. A state $s = (l_1, \dots, l_k)$ represents the coordinate positions of all robots.

\paragraph{Dynamics.} In each round, every robot independently chooses a move from $A = \{\text{N}, \text{S}, \text{E}, \text{W}, \text{Wait}\}$. Simultaneously, the opponent sets the wind condition from $A_{\opp} = \{\text{N}, \text{S}, \text{E}, \text{W}, \text{Calm}\}$.

\paragraph{Transitions.} The game includes a time penalty to motivate the team to reach the target quickly. At the end of each round, with probability $p=0.5$, the game transitions to a sink loss state. If the game continues, transitions depend on the interaction between the robot's action and the wind. If any two robots occupy the same cell, the team collides and loses immediately. Otherwise, movement occurs as follows:
\begin{enumerate}
    \item \textbf{Calm conditions:} If the wind is \textit{Calm}, all robots move deterministically to their chosen target (or remain stationary if they chose \textit{Wait}).
    \item \textbf{Waiting in wind:} If a robot chooses \textit{Wait} during a storm, the wind pushes it one step in the wind's direction with probability 1.
    \item \textbf{Moving in wind:} If a robot attempts to move, we compare its direction to the wind:
    \begin{itemize}
        \item \textbf{Tailwind:} If the directions match, the wind assists the robot. It moves to the target cell with probability 1.
        \item \textbf{Headwind or crosswind:} If the robot fights the wind, it reaches its target with probability 0.5. With probability 0.5, it is overpowered and pushed one step in the wind's direction.
    \end{itemize}
\end{enumerate}

\paragraph{Objective.} The team aims to maximise the probability of reaching a target configuration $T$ without collisions and before the random termination occurs. The opponent minimises this probability.

We evaluated four scenarios shown in Figure \ref{fig:robot_scenarios}. The results are presented in Table \ref{table:robotgame}.

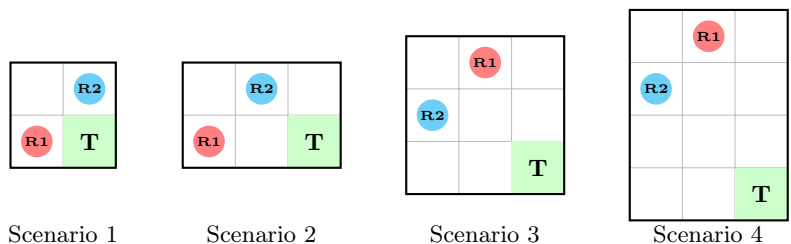
\begin{figure*}[h!]
\centering
\setlength{\tabcolsep}{12pt}
\begin{tabular}{cccc}
\begin{tikzpicture}[scale=0.7, baseline=(current bounding box.center)]
  \draw[step=1.0,gray!50,thin] (0,0) grid (2,2);
  \fill[green!20] (1.0,0.0) rectangle ++(1,1);
  \node[font=\footnotesize] at (1.5,0.5) {\textbf{T}};
  \fill[red!50] (0.5,0.5) circle (0.3);
  \node[font=\tiny] at (0.5,0.5) {\textbf{R1}};
  \fill[cyan!50] (1.5,1.5) circle (0.3);
  \node[font=\tiny] at (1.5,1.5) {\textbf{R2}};
  \draw[thick] (0,0) rectangle (2,2);
\end{tikzpicture} & \begin{tikzpicture}[scale=0.7, baseline=(current bounding box.center)]
  \draw[step=1.0,gray!50,thin] (0,0) grid (3,2);
  \fill[green!20] (2.0,0.0) rectangle ++(1,1);
  \node[font=\footnotesize] at (2.5,0.5) {\textbf{T}};
  \fill[red!50] (0.5,0.5) circle (0.3);
  \node[font=\tiny] at (0.5,0.5) {\textbf{R1}};
  \fill[cyan!50] (1.5,1.5) circle (0.3);
  \node[font=\tiny] at (1.5,1.5) {\textbf{R2}};
  \draw[thick] (0,0) rectangle (3,2);
\end{tikzpicture} & \begin{tikzpicture}[scale=0.7, baseline=(current bounding box.center)]
  \draw[step=1.0,gray!50,thin] (0,0) grid (3,3);
  \fill[green!20] (2.0,0.0) rectangle ++(1,1);
  \node[font=\footnotesize] at (2.5,0.5) {\textbf{T}};
  \fill[red!50] (1.5,2.5) circle (0.3);
  \node[font=\tiny] at (1.5,2.5) {\textbf{R1}};
  \fill[cyan!50] (0.5,1.5) circle (0.3);
  \node[font=\tiny] at (0.5,1.5) {\textbf{R2}};
  \draw[thick] (0,0) rectangle (3,3);
\end{tikzpicture} & \begin{tikzpicture}[scale=0.7, baseline=(current bounding box.center)]
  \draw[step=1.0,gray!50,thin] (0,0) grid (3,4);
  \fill[green!20] (2.0,0.0) rectangle ++(1,1);
  \node[font=\footnotesize] at (2.5,0.5) {\textbf{T}};
  \fill[red!50] (1.5,3.5) circle (0.3);
  \node[font=\tiny] at (1.5,3.5) {\textbf{R1}};
  \fill[cyan!50] (0.5,2.5) circle (0.3);
  \node[font=\tiny] at (0.5,2.5) {\textbf{R2}};
  \draw[thick] (0,0) rectangle (3,4);
\end{tikzpicture} \\
\footnotesize Scenario 1 & \footnotesize Scenario 2 & \footnotesize Scenario 3 & \footnotesize Scenario 4 \\
\end{tabular}
\caption{The four benchmark scenarios for the robot game. Robots are shown as circles ($R_i$), and the target is marked with T.}
\label{fig:robot_scenarios}
\end{figure*}

\subsubsection{Jamming multi-channel radio systems.} Adapted from \cite{kwiatkowska2021automatic}, this benchmark models $k$ sensors $\Players = \{x_1, \dots, x_k\}$ transmitting packets over $C$ frequency channels against a jammer $\opp$. A state is defined by buffer levels $s = (b_1, \dots, b_k)$. Initially, each sensor $x_i$ holds $B_i$ packets; $b_i \in \{0, \dots, B_i\}$ tracks the remaining packets.

\textit{Dynamics.} In each round, sensors independently choose an action from $A = \{1, \dots, C, \text{Wait}\}$, while the jammer targets a channel from $A_{\opp} = \{1, \dots, C, \text{Idle}\}$. Transition rules for sensor $x_i$ on channel $c$:

\begin{enumerate}
    \item \textbf{Successful transmission}: The buffer decrements ($b'_i = b_i - 1$) if and only if the opponent does not jam channel $c$ ($a_{\opp} \neq c$) and no other sensor selects $c$ (no collision).
    \item \textbf{Transmission failure}: If channel $c$ is jammed by $\opp$ or selected by another sensor (collision), the team players lose.
    \item \textbf{Wait}: If the sensor chooses \textit{Wait}, it makes no transmission attempt, and the buffer remains unchanged ($b'_i = b_i$).
\end{enumerate}

\textbf{Objective}. The team maximises the probability of reaching the target $T = (0, \dots, 0)$ where all buffers are empty, while the opponent minimises this probability. 

The results for this benchmark are presented in Table \ref{table:jamminggame}.

\subsection{Results}

\label{app:experiment_results}

This section contains the full results for each of the benchmarks. They are presented in Tables \ref{table:graphgame}, \ref{table:robotgame}, \ref{table:jamminggame}, and \ref{table:almostsure} in the following pages.
\newpage

\begin{table}[h]
\centering
\small 
\begin{tabular}{|c|c|c|c|c|c|c|c|c|c|}
\hline
\multirow{2}{*}{\textbf{Scenario}} & \multirow{2}{*}{\textbf{Team}} & \textbf{States} & \multirow{2}{*}{\textbf{Algorithm}} & \multicolumn{3}{c|}{\textbf{Individual Randomness}} & \multicolumn{3}{c|}{\textbf{Shared Randomness}} \\
 & & \textbf{Trans} & & \textbf{Time} & \textbf{Val} & \textbf{Iter} & \textbf{Time} & \textbf{Val} & \textbf{Iter} \\ \hline
\multirow{5}{*}{1} & \multirow{5}{*}{2} & \multirow{5}{*}{\shortstack{27\\512}} & ETR-Direct & T.O. & N/F & N/A & T.O. & N/F & N/A \\
 & & & VI-ETR & 4.42 & 0.290 & 6 & 0.40 & 0.500 & 2 \\
 & & & VI-OPT & 0.54 & 0.290 & 6 & 0.02 & 0.500 & 2 \\
 & & & VI-Hybrid & 0.58 & 0.290 & 6 & 0.05 & 0.500 & 2 \\
 & & & PRISM & - & - & N/A & 1.58 & 0.500 & N/A \\ \hline
\multirow{5}{*}{2} & \multirow{5}{*}{2} & \multirow{5}{*}{\shortstack{64\\4096}} & ETR-Direct & T.O. & N/F & N/A & T.O. & N/F & N/A \\
 & & & VI-ETR & T.O. & 0.000 & 1 & 3.83 & 0.750 & 2 \\
 & & & VI-OPT & 0.39 & 0.300 & 9 & 0.27 & 0.750 & 2 \\
 & & & VI-Hybrid & T.O. & 0.000 & 1 & 0.56 & 0.750 & 2 \\
 & & & PRISM & - & - & N/A & 1.58 & 0.750 & N/A \\ \hline
\multirow{5}{*}{3} & \multirow{5}{*}{2} & \multirow{5}{*}{\shortstack{64\\729}} & ETR-Direct & T.O. & N/F & N/A & 0.48 & 1.000 & N/A \\
 & & & VI-ETR & 14.38 & 0.296 & 6 & 47.41 & 0.990 & 98 \\
 & & & VI-OPT & 0.23 & 0.296 & 6 & 2.27 & 0.990 & 100 \\
 & & & VI-Hybrid & 1.77 & 0.296 & 6 & 6.12 & 0.990 & 100 \\
 & & & PRISM & - & - & N/A & 1.34 & 0.999 & N/A \\ \hline
\multirow{5}{*}{4} & \multirow{5}{*}{2} & \multirow{5}{*}{\shortstack{125\\1331}} & ETR-Direct & T.O. & N/F & N/A & T.O. & N/F & N/A \\
 & & & VI-ETR & 35.21 & 0.348 & 4 & 76.80 & 0.993 & 70 \\
 & & & VI-OPT & 0.27 & 0.348 & 4 & 3.66 & 0.993 & 72 \\
 & & & VI-Hybrid & 2.06 & 0.348 & 4 & 10.19 & 0.993 & 72 \\
 & & & PRISM & - & - & N/A & 1.36 & 0.999 & N/A \\ \hline
\multirow{5}{*}{5} & \multirow{5}{*}{3} & \multirow{5}{*}{\shortstack{256\\65536}} & ETR-Direct & T.O. & N/F & N/A & T.O. & N/F & N/A \\
 & & & VI-ETR & T.O. & 0.000 & 1 & 57.89 & 0.750 & 2 \\
 & & & VI-OPT & 1.98 & 0.075 & 8 & 3.00 & 0.750 & 2 \\
 & & & VI-Hybrid & T.O. & 0.000 & 1 & 15.49 & 0.750 & 2 \\
 & & & PRISM & - & - & N/A & 7.32 & 0.750 & N/A \\ \hline
\multirow{5}{*}{6} & \multirow{5}{*}{3} & \multirow{5}{*}{\shortstack{1296\\38416}} & ETR-Direct & T.O. & N/F & N/A & T.O. & N/F & N/A \\
 & & & VI-ETR & T.O. & 0.125 & 2 & T.O. & 0.976 & 22 \\
 & & & VI-OPT & 4.38 & 0.173 & 3 & 138.64 & 0.993 & 72 \\
 & & & VI-Hybrid & T.O. & 0.125 & 2 & 312.18 & 0.993 & 72 \\
 & & & PRISM & - & - & N/A & 2.29 & 0.999 & N/A \\ \hline
\end{tabular}
\caption{Experimental results for the six scenarios illustrated in Figure \ref{fig:gamegraphs}. For each algorithm, the table lists the execution time in seconds, the computed value, and the number of iterations required for convergence. The label ``T.O.'' denotes a timeout after 600 seconds. In cases where an algorithm timed out, we report the value obtained from the final completed iteration. ``N/F'' indicates Not Found.}
\label{table:graphgame}
\end{table}

\begin{table}
\centering
\footnotesize
\setlength{\tabcolsep}{3pt} 
\begin{tabular}{|c|c|c|c|c|c|c|c|c|c|}
\hline
\multirow{2}{*}{\textbf{Scenario}} & \multirow{2}{*}{\textbf{Team Size}} & \textbf{States} & \multirow{2}{*}{\textbf{Algorithm}} & \multicolumn{3}{c|}{\textbf{Individual Randomness}} & \multicolumn{3}{c|}{\textbf{Shared Randomness}} \\
 & & \textbf{Trans} & & \textbf{Time} & \textbf{Value} & \textbf{Iters} & \textbf{Time} & \textbf{Value} & \textbf{Iters} \\ \hline
\multirow{5}{*}{1} & \multirow{5}{*}{2} & \multirow{5}{*}{\shortstack{16\\2000}} & ETR-Direct & T.O. & N/F & N/A & T.O. & N/F & N/A \\
 & & & VI-ETR & 7.86 & 0.293 & 6 & 4.86 & 0.324 & 6 \\
 & & & VI-OPT & 0.56 & 0.274 & 6 & 0.72 & 0.323 & 6 \\
 & & & VI-Hybrid & 5.35 & 0.293 & 6 & 2.27 & 0.324 & 6 \\
 & & & PRISM & - & - & N/A & 1.47 & 0.324 & N/A \\ \hline
\multirow{5}{*}{2} & \multirow{5}{*}{2} & \multirow{5}{*}{\shortstack{36\\4500}} & ETR-Direct & T.O. & N/F & N/A & T.O. & N/F & N/A \\
 & & & VI-ETR & T.O. & 0.094 & 3 & 23.31 & 0.115 & 8 \\
 & & & VI-OPT & 1.33 & 0.113 & 8 & 2.81 & 0.115 & 8 \\
 & & & VI-Hybrid & T.O. & 0.000 & 2 & 13.60 & 0.115 & 8 \\
 & & & PRISM & - & - & N/A & 1.74 & 0.115 & N/A \\ \hline
\multirow{5}{*}{3} & \multirow{5}{*}{2} & \multirow{5}{*}{\shortstack{81\\10125}} & ETR-Direct & T.O. & N/F & N/A & T.O. & N/F & N/A \\
 & & & VI-ETR & T.O. & 0.000 & 3 & 68.28 & 0.038 & 9 \\
 & & & VI-OPT & 3.61 & 0.033 & 9 & 9.54 & 0.036 & 9 \\
 & & & VI-Hybrid & T.O. & 0.000 & 3 & 42.59 & 0.038 & 9 \\
 & & & PRISM & - & - & N/A & 2.56 & 0.038 & N/A \\ \hline
\multirow{5}{*}{4} & \multirow{5}{*}{2} & \multirow{5}{*}{\shortstack{144\\18000}} & ETR-Direct & T.O. & N/F & N/A & T.O. & N/F & N/A \\
 & & & VI-ETR & T.O. & 0.000 & 3 & 143.59 & 0.011 & 10 \\
 & & & VI-OPT & 6.09 & 0.004 & 9 & 17.10 & 0.010 & 10 \\
 & & & VI-Hybrid & T.O. & 0.000 & 3 & 96.53 & 0.011 & 10 \\
 & & & PRISM & - & - & N/A & 3.64 & 0.011 & N/A \\ \hline
\end{tabular}
\caption{Experimental results for the four scenarios illustrated in Figure 
\ref{fig:robot_scenarios}. For each algorithm, the table lists the execution 
time in seconds, the computed value, and the number of iterations 
required for convergence. The label "T.O." denotes a timeout after 600 seconds. 
In cases where an algorithm timed out, we report the value obtained 
from the final completed iteration.}
\label{table:robotgame}
\end{table}

\begin{table}
\centering
\footnotesize
\setlength{\tabcolsep}{2.5pt} 
\begin{tabular}{|c|c|c|c|c|c|c|c|c|}
\hline
\multirow{2}{*}{\textbf{Parameters}} & \textbf{States} & \multirow{2}{*}{\textbf{Algorithm}} & \multicolumn{3}{c|}{\textbf{Individual Randomness}} & \multicolumn{3}{c|}{\textbf{Shared Randomness}} \\
 & \textbf{Trans} & & \textbf{Time} & \textbf{Value} & \textbf{Iters} & \textbf{Time} & \textbf{Value} & \textbf{Iters} \\ \hline
\multirow{5}{*}{C=2, B=[1, 1]} & \multirow{5}{*}{\shortstack{5\\135}} & ETR-Direct & T.O. & N/F & N/A & T.O. & N/F & N/A \\
 & & VI-ETR & 0.33 & 0.250 & 3 & 0.25 & 0.250 & 3 \\
 & & VI-OPT & 0.62 & 0.246 & 46 & 0.02 & 0.250 & 3 \\
 & & VI-Hybrid & 0.08 & 0.250 & 3 & 0.03 & 0.250 & 3 \\
 & & PRISM & - & - & N/A & 1.38 & 0.250 & N/A \\ \hline
\multirow{5}{*}{C=2, B=[2, 2]} & \multirow{5}{*}{\shortstack{10\\270}} & ETR-Direct & T.O. & N/F & N/A & T.O. & N/F & N/A \\
 & & VI-ETR & 1.35 & 0.062 & 5 & 0.96 & 0.062 & 5 \\
 & & VI-OPT & 1.23 & 0.061 & 46 & 0.07 & 0.063 & 5 \\
 & & VI-Hybrid & 0.31 & 0.063 & 5 & 0.14 & 0.063 & 5 \\
 & & PRISM & - & - & N/A & 1.29 & 0.062 & N/A \\ \hline
\multirow{5}{*}{C=3, B=[3, 3]} & \multirow{5}{*}{\shortstack{17\\1088}} & ETR-Direct & T.O. & N/F & N/A & T.O. & N/F & N/A \\
 & & VI-ETR & T.O. & 0.000 & 2 & 6.18 & 0.088 & 7 \\
 & & VI-OPT & 2.67 & 0.084 & 39 & 0.50 & 0.088 & 7 \\
 & & VI-Hybrid & T.O. & 0.000 & 1 & 1.15 & 0.088 & 7 \\
 & & PRISM & - & - & N/A & 1.37 & 0.088 & N/A \\ \hline
\multirow{5}{*}{C=3, B=[4, 4]} & \multirow{5}{*}{\shortstack{26\\1664}} & ETR-Direct & T.O. & N/F & N/A & T.O. & N/F & N/A \\
 & & VI-ETR & T.O. & 0.000 & 2 & 11.53 & 0.039 & 9 \\
 & & VI-OPT & 4.34 & 0.039 & 39 & 0.83 & 0.039 & 9 \\
 & & VI-Hybrid & T.O. & 0.000 & 1 & 2.31 & 0.039 & 9 \\
 & & PRISM & - & - & N/A & 1.39 & 0.039 & N/A \\ \hline
\multirow{5}{*}{C=4, B=[5, 5]} & \multirow{5}{*}{\shortstack{37\\4625}} & ETR-Direct & T.O. & N/F & N/A & T.O. & N/F & N/A \\
 & & VI-ETR & T.O. & 0.000 & 1 & 34.32 & 0.056 & 11 \\
 & & VI-OPT & 8.24 & 0.050 & 40 & 2.62 & 0.056 & 11 \\
 & & VI-Hybrid & T.O. & 0.000 & 1 & 9.41 & 0.056 & 11 \\
 & & PRISM & - & - & N/A & 1.44 & 0.056 & N/A \\ \hline
\multirow{5}{*}{C=4, B=[6, 6]} & \multirow{5}{*}{\shortstack{50\\6250}} & ETR-Direct & T.O. & N/F & N/A & T.O. & N/F & N/A \\
 & & VI-ETR & T.O. & 0.000 & 1 & 55.06 & 0.032 & 13 \\
 & & VI-OPT & 10.24 & 0.024 & 40 & 4.59 & 0.030 & 14 \\
 & & VI-Hybrid & T.O. & 0.000 & 1 & 17.41 & 0.032 & 13 \\
 & & PRISM & - & - & N/A & 1.54 & 0.032 & N/A \\ \hline
\end{tabular}
\caption{Performance comparison of algorithms on the jamming multi-channel radio system 
benchmark. Each scenario is defined by the number of channels $C$ followed by 
the initial buffer sizes $[B_1, \dots, B_k]$ for $k$ players. We report the 
number of states, number of transitions, execution time in seconds, computed 
game value $V$, and the number of iterations for the value iteration algorithms. ``T.O.'' indicates the algorithm exceeded 
the 600s time limit. For instances that timed out, we provide the value 
from the last completed iteration, where available.}
\label{table:jamminggame}
\end{table}

\begin{table}
\centering
\footnotesize
\setlength{\tabcolsep}{2.5pt} 
\begin{tabular}{|c|c|c|c|c|c|c|c|}
\hline
\multirow{2}{*}{\textbf{Scenario}} & \multirow{2}{*}{\begin{tabular}[c]{@{}c@{}}\textbf{States} \\ \textbf{Trans}\end{tabular}} & \multirow{2}{*}{\textbf{Method}} & \multicolumn{2}{c|}{\textbf{Individual Randomness}} & \multicolumn{2}{c|}{\textbf{Shared Randomness}} \\ \cline{4-7}
 & & & \textbf{Time} & \textbf{Result} & \textbf{Time} & \textbf{Result} \\ \hline
\multirow{2}{*}{1} & \multirow{2}{*}{\begin{tabular}[c]{@{}c@{}} 27 \\ 729 \end{tabular}} & SAT-Direct & 0.02 & UNSAT & 0.02 & UNSAT \\
 & & PRISM-Qual & -- & -- & 1.62 & UNSAT \\ \hline
\multirow{2}{*}{2} & \multirow{2}{*}{\begin{tabular}[c]{@{}c@{}} 64 \\ 1728 \end{tabular}} & SAT-Direct & 0.09 & SAT & 0.09 & SAT \\
 & & PRISM-Qual & -- & -- & 1.42 & SAT \\ \hline
\multirow{2}{*}{3} & \multirow{2}{*}{\begin{tabular}[c]{@{}c@{}} 125 \\ 4913 \end{tabular}} & SAT-Direct & 0.29 & UNSAT & 0.33 & UNSAT \\
 & & PRISM-Qual & -- & -- & 1.39 & UNSAT \\ \hline
\multirow{2}{*}{4} & \multirow{2}{*}{\begin{tabular}[c]{@{}c@{}} 216 \\ 5832 \end{tabular}} & SAT-Direct & 0.67 & SAT & 0.60 & SAT \\
 & & PRISM-Qual & -- & -- & 1.43 & SAT \\ \hline
\multirow{2}{*}{5} & \multirow{2}{*}{\begin{tabular}[c]{@{}c@{}} 512 \\ 21952 \end{tabular}} & SAT-Direct & 5.26 & SAT & 8.46 & SAT \\
 & & PRISM-Qual & -- & -- & 2.94 & SAT \\ \hline
\multirow{2}{*}{6} & \multirow{2}{*}{\begin{tabular}[c]{@{}c@{}} 1000 \\ 46656 \end{tabular}} & SAT-Direct & 20.54 & UNSAT & 19.93 & UNSAT \\
 & & PRISM-Qual & -- & -- & 5.98 & UNSAT \\ \hline
\multirow{2}{*}{7} & \multirow{2}{*}{\begin{tabular}[c]{@{}c@{}} 1728 \\ 97336 \end{tabular}} & SAT-Direct & 93.70 & UNSAT & 58.45 & UNSAT \\
 & & PRISM-Qual & -- & -- & 42.27 & UNSAT \\ \hline
\end{tabular}
\caption{Experimental results for the almost-sure reachability problem on the Pursuit-Evasion benchmark. The scenarios are illustrated in Figure \ref{fig:almost_sure_graph}. For each scenario, the table lists the execution time in seconds and the outcome. We compare our SAT-Direct algorithm (under both individual and shared randomness) against the PRISM-games baseline (which supports only shared randomness). A result of \textbf{SAT} indicates that the target is almost-surely reachable, while \textbf{UNSAT} indicates that it is not.}
\label{table:almostsure}
\end{table}

\pagebreak

\end{document}